\documentclass[showpacs,
amsmath,
amssymb,
onecolumn,superscriptaddress,notitlepage,preprintnumbers,pra]{revtex4-2}
\usepackage[dvips]{graphicx}
\usepackage{amsmath,amssymb,amsthm,mathrsfs,amsfonts,dsfont}
\usepackage{subfigure, epsfig}
\usepackage{braket}
\usepackage{bm}
\usepackage{enumerate}
\usepackage{color,xcolor}
\usepackage{physics}
\usepackage{float}
\usepackage{graphicx}
\usepackage{hyperref}
\hypersetup{colorlinks=true, linkcolor=blue, citecolor=blue, urlcolor=black }
\usepackage{comment}
\usepackage[normalem]{ulem}
\usepackage{booktabs}
\usepackage{mathtools}
\usepackage{lineno}
\usepackage[version=4]{mhchem}
\usepackage[ruled, vlined, linesnumbered, commentsnumbered, longend]{algorithm2e}
\newcommand{\grads}{\bm{g}}

\usepackage{varioref}
\labelformat{equation}{(#1)}

\labelformat{section}{#1}

\labelformat{figure}{#1}

\labelformat{proposition}{#1}

\labelformat{lemma}{#1}

\labelformat{theorem}{#1}

\labelformat{observation}{#1}

\labelformat{definition}{#1}

\labelformat{problem}{#1}

\labelformat{algorithm}{#1}

\labelformat{corollary}{#1}

\labelformat{statement}{#1}

\newtheorem{theorem}{Theorem}

\begin{document}
 
\title{Enhanced Neural Quantum State via Annealed Gradient Descent}
 
\date{\today}


\author{Shiwei Zhou}
\thanks{These authors contributed equally to this work.}
\affiliation{Institute of High Energy Physics, Chinese Academy of Sciences, Beijing 100049, China}%

\author{Yiming Huang}
\thanks{These authors contributed equally to this work.}
\affiliation{Institute of High Energy Physics, Chinese Academy of Sciences, Beijing 100049, China}
\affiliation{China Center of Advanced Science and Technology, Beijing 100190, China}
\affiliation{Center on Frontiers of Computing Studies, School of Computer Science, Peking University, Beijing 100871, China}

\author{Xiao Yuan}
\email{xiaoyuan@pku.edu.cn}
\affiliation{Center on Frontiers of Computing Studies, School of Computer Science, Peking University, Beijing 100871, China}

\author{Xiaoxia Cai}
\email{xxcai@ihep.ac.cn}
\affiliation{Institute of High Energy Physics, Chinese Academy of Sciences, Beijing 100049, China}

\begin{abstract}

Neural quantum states offer expressive representations of quantum many-body wave functions, yet their practical accuracy can be limited by stochastic optimization rather than representational capacity. Here we identify a finite-sample instability, termed subspace trapping, in which physically important configurations become strongly underestimated, remain absent from successive sampling batches and receive insufficient gradient feedback. 
This self-reinforcing loss of sampled support can confine optimization to an effective subspace and produce apparently stationary states above the true ground state energy. 
To address this problem, we introduce annealed gradient descent (AGD), a sampling-aware update with annealing factor that temporarily increases the relative contribution of sampled low-probability configurations while limiting the dominance of high-probability ones. 
We establish the connection between finite-sample support loss and effective subspace optimization, and then evaluate the method across molecular systems, one and two-dimensional $J_1$-$J_2$ models. Annealed gradient descent suppresses metastable trapping, preserves physically relevant configurations and enables compact neural quantum states to attain chemical accuracy and competitive state-of-the-art performance. 
These results establish AGD as a lightweight complement to expressive neural architectures, improved sampling strategies for scalable quantum many-body optimization.

\end{abstract}

\maketitle

\let\oldaddcontentsline\addcontentsline
\renewcommand{\addcontentsline}[3]{}

\section{Introduction}
Simulating quantum many-body systems is a central challenge in computational physics and chemistry. Exact wave-function representations are limited by the exponential growth of Hilbert space, motivating variational approaches that restrict the search to a parameterized family of states and optimize its free parameters \cite{peruzzo2014variational,carleo2017solving}. The success of such methods therefore depends not only on the expressive capacity of the variational ansatz, but also on whether the optimization procedure can reliably navigate the resulting high-dimensional landscape. Neural quantum states (NQS) have emerged as a powerful variational framework in this setting \cite{carleo2017solving}. By using neural networks to parameterize quantum amplitudes, NQS provide flexible representations that can extend beyond conventional tensor-network descriptions and, in favourable settings, encode strongly correlated states with high-dimensional or volume-law entanglement \cite{deng2017quantum,gao2017efficient}. This flexibility has motivated a broad range of architectures, including restricted Boltzmann machines, autoregressive models, recurrent neural networks and Transformer-based wave functions \cite{carleo2017solving,sharir2020deep,
hibatallah2020recurrent,barrett2022autoregressive,
chung2023nnqs,zhang2023transformer,wang2025solving,shang2025solving,
zeng2026transformer,sun2026fully,
kan2025accelerating,fu2025evaluation}, and has enabled applications to quantum many-body physics and electronic-structure problems \cite{lange2024architectures,medvidovic2024neural,hermann2020deep,pfau2020abinitio,choo2020fermionic}.

However, the practical accuracy of NQS is determined not only by what the network can represent, but also by which regions of configuration space remain visible to the stochastic optimization process. In variational Monte Carlo, observables and gradients are estimated from configurations sampled from the current model distribution. Near critical points and in strongly correlated regimes, Markov-chain sampling can exhibit long autocorrelation times and incomplete mixing, leading to noisy or biased finite-sample estimates. At the same time, the neural parameterization maps the ground-state problem onto a highly non-convex optimization landscape, in which standard first-order methods may converge to metastable states that are numerically stable but remain above the true ground-state energy \cite{lange2024architectures,medvidovic2024neural}. Existing methods address different parts of this computational pipeline. Stochastic reconfiguration and natural-gradient methods improve the geometry of parameter updates by approximating imaginary-time evolution on the variational manifold \cite{sorella1998green,sorella2001generalized}, while scalable formulations such as MinSR, related linear-algebra reformulations, KFAC and SPRING reduce the computational cost of curvature-aware optimization for large neural wave functions \cite{chen2024empowering,rende2024simple,martens2015optimizing,goldshlager2024kaczmarz}. Complementary sampling-oriented approaches use importance sampling, neural importance resampling, sampling without replacement or deterministic configuration selection to reduce estimator variance and improve the representation of peaked or multimodal distributions \cite{misery2026looking,ledinauskas2025neural,malyshev2024neural,li2023nonstochastic}.

Despite these advances, a central question remains unresolved: \textit{why can an expressive NQS converge to an apparently stationary solution without recovering the true ground state, and why does escaping such a solution often require prohibitively large sample sizes?} One possible explanation is the difficulty of learning nontrivial sign or phase structures, particularly in frustrated systems, where optimization may converge to a low-energy state with an incorrect sign pattern \cite{Szabo2020sign}. Rugged optimization landscapes and mode-collapse-like behaviour have also been identified as potential failure mechanisms \cite{Wu2019VAN,Wu2021cluster}. These effects are closely coupled to amplitude optimization, because incorrect phase relations modify the energy gradient and can indirectly suppress the amplitudes of physically important configurations. Several strategies have been developed to mitigate such failures, including temperature or variational annealing \cite{wu2019solving,hibat2021variational}, symmetry restoration and neural cluster updates \cite{wu2021unbiased}, and collective-variable penalties that prevent unphysical concentration of probability mass \cite{zhang2023understanding}. However, these approaches may still recover only a subset of degenerate modes, while neural samplers can remain difficult to train for intrinsically multimodal distributions \cite{inack2022neural,ciarella2023machine}. More fundamentally, existing explanations do not fully clarify how finite-sample optimization dynamics are altered once physically important configurations become strongly suppressed and therefore rarely, if ever, sampled.

In this work, we firstly identify a self-reinforcing finite-sample instability in NQS optimization, which we term subspace trapping. The instability emerges when physically important configurations are assigned probabilities below the effective resolution of the sampling budget. Such configurations may then remain absent from many consecutive batches and contribute no direct statistical weight to either the empirical energy or covariance based gradient. Although its amplitude may still enter local-energy evaluations indirectly through sampled configurations connected to it by the Hamiltonian, the optimizer receives little direct information with which to restore its probability. Because every parameter update reshapes the distribution from which subsequent samples are drawn, configurations that are initially underestimated receive progressively weaker corrective feedback, whereas those that produce an immediate reduction in the estimated energy are preferentially reinforced. This coupling between sampling and optimization can confine the dynamics to an effectively accessible subspace of Hilbert space, yielding apparent convergence at an energy above the true ground-state energy. Increasing the batch size lowers the probability scale at which configurations become unresolved, but does not remove the instability, because the evolving model can subsequently suppress other relevant configurations below the new sampling threshold.
To address this issue, we then introduce annealed gradient descent (AGD), a sampling-aware update rule that applies a temperature-dependent reweighting to configuration-level gradient contributions. Early in the optimization, AGD enhances the relative influence of sampled low-probability configurations while limiting the premature dominance of high-probability ones, thereby maintaining broader support before relevant configurations become statistically inaccessible. The correction is then gradually annealed away, recovering the standard variational gradient for final energy refinement. AGD is conceptually distinct from natural-gradient methods, MinSR and SPRING, because it neither constructs nor approximates a parameter-space metric. It also differs from importance-sampling and resampling strategies, as it leaves the sampling distribution unchanged and instead reweights the information already contained in each batch. Consequently, AGD requires neither architectural modifications nor additional curvature estimates, while preserving the leading-order computational complexity of the underlying first-order optimizer.

We combine AGD with a theoretical analysis of the relationship between finite-sample support loss and effective subspace optimization. We show how the persistent suppression of a small set of configurations can produce an optimization trajectory that is nearly stationary within the accessible sampling subspace while remaining inconsistent with the ground state of the full Hamiltonian. We then validate the method across molecular systems and one and two-dimensional $J_1$-$J_2$ lattice models, covering qualitatively distinct Hamiltonians, correlation regimes and wave-function structures. Across these benchmarks, AGD prevents physically important configurations from becoming persistently undersampled, avoids convergence to metastable solutions such as the Hartree-Fock state, and enables comparatively simple NQS models to attain chemical accuracy in molecular calculations. These results establish finite-sample support loss as a distinct source of failure in self-sampled NQS optimization and show that gradient-level annealing provides a lightweight complement to expressive neural architectures, improved sampling strategies and geometry-aware optimizers.

\section{Results}
\label{sec:Results}
In this section, we first provide a theoretical analysis of the optimization difficulties that arise in NQS training, then introduce AGD as a training strategy to mitigate these issues, and finally evaluate its effectiveness through numerical results on both molecular and many-body systems.

\subsection{The analysis of NQS training issue}
Here, we first recall that given a Hamiltonian $H$, within the framework of variational Monte Carlo (VMC), estimating the energy gradient $\grads = \nabla \langle H\rangle$ with respect to the variational parameters $\bm{\theta}\in\mathbb{R}^{N_p}$ requires a set $\mathcal{S}=\{\sigma^{(k)}\}_{k=1}^{N_S}$ of $N_S$ samples drawn from the probability distribution $p_{\bm{\theta}}(\sigma)=|\psi_{\bm{\theta}}(\sigma)|^2$, where $\psi_{\bm{\theta}}(\sigma)$ denotes the neural network wave function associated with the configuration $\sigma$. In particular, the $g_i$ is given by
\begin{align}
    g_i &= \mathbb{E}_{\sigma \sim p_{\bm{\theta}}}\Big[2\Re\left[\partial_{\theta_i} \log \psi_{\bm{\theta}}^{*}(\sigma)\left(E_{\mathrm{loc}}(\sigma) - E\right)\Big]\right]\\
    &\approx \frac{2}{N_S}\Re\left[\sum_{\sigma\in\mathcal{S}}\partial_{\theta_i} \log
\psi_{\bm{\theta}}^{*}(\sigma)\left(E_{\mathrm{loc}}(\sigma) - E\right)\right],
    \label{eq:vmc_gradient}
\end{align}
where $E_{\mathrm{loc}}(\sigma)=\sum_{\sigma'} H_{\sigma \sigma'}\psi_{\bm{\theta}}(\sigma')/\psi_{\bm{\theta}}(\sigma)$ refers to local energy and $E = \mathbb{E}_{\sigma \in \mathcal{S}}[E_{\mathrm{loc}}(\sigma)]$ denotes the variational energy. Accordingly, under the gradient descent scheme, the parameter update takes the form $\bm{\theta}\leftarrow \bm{\theta} - \eta\grads$, where $\eta$ is the step size.

However, this framework may suffers from a severe sampling bias problem, sometimes referred to as mode collapse. As the system size grows, the Hilbert space expands exponentially, while the NQS ansatz estimates energy gradients $\grads$ from a limited number of configurations $\mathcal{S}$ sampled according to the model distribution $p_{\bm{\theta}}$, which is induced by the neural network wave function $\psi_{\bm{\theta}}(\sigma)$.
This gives rise to a self-consistent sampling dilemma. When the current parameters $\bm{\theta}$ assign insufficient probability to configurations that make essential contributions to the ground state, those configurations are rarely sampled. As a result, their contributions to the gradient are statistically suppressed, preventing the optimization from steering the parameters toward the physically correct region of the Hilbert space.
Such problem is especially pronounced in strongly correlated complex systems, where the ground state typically involves a superposition of many configurations with comparable amplitudes. As the system size increases, the number of relevant configurations can grow rapidly, and their support may become broadly distributed across the Hilbert space, making them difficult to cover with a finite sampling budget. The following theorem formalizes the mechanism underlying this problem and the proof are deferred to Appendix.~\ref{Appendix:proof of thm:Hsubspace}.

\begin{theorem}
    \label{thm:Hsubspace}
    Assuming the state $\sigma_0$ is never sampled, that is over the $\Delta t$ optimization steps and $p_{\bm{\theta}}(\sigma_0) \ll \epsilon$, where $\epsilon=1/\Delta tN_s$ so that $p_{\bm{\theta}}(\sigma_0)$ can be effectively treated as zero throughout these steps. The gradient-based optimization process is mathematically equivalent to finding the exact eigenstate of a subspace Hamiltonian $\hat{H}$, which inevitably leads to the change of $p_{\bm{\theta}}(\sigma_0)$ close to 0, that is $\Delta p_{\bm{\theta}}(\sigma_0)=0$, during the optimization within this subspace.
\end{theorem}
The results suggest that, during the optimization process, if a critical configuration $\sigma_0$ becomes incorrectly suppressed to the extent that $p_{\bm{\theta}}(\sigma_0) \ll \epsilon$, its contribution to the local energy and gradient becomes statistically invisible, and the sampling-based approach consequently fails to capture this essential component.
By Theorem~\ref{thm:Hsubspace}, this causes the probability $p_{\bm{\theta}}(\sigma_0)$ to remain nearly fixed and insensitive to subsequent gradient updates, which can no longer reliably restore the missing component. 
Consequently, the optimization no longer targets the true ground state, instead, it converges to the ground state of an effective Hamiltonian projected onto the restricted subspace $\mathcal{H}_S$ that excludes $\sigma_0$, trapping the NQS as illustrated in Fig.~\ref{fig:landscape}. 
This explains how an NQS can satisfy apparent convergence criteria while retaining an energy above the true ground-state energy, and also clarifies and why brute-force increases in sample size can partially help by enlarging the accessible subspace and recovering configurations that were previously invisible.
Within the VMC framework, the optimization only evluate the local energy in Eq.~\ref{eq:local_energy}, rather than as a dynamical generator acting on the full many-body state. 
For each sampled configuration $\sigma$, Hamiltonian $H$ provides the information in $E_{\mathrm{loc}}(\sigma)$, determined by the configurations locally connected to $\sigma$ through its sparse structure. 
As a result, the algorithm has no direct access to the global spectrum of $H$, information about the ground state is inferred only indirectly from pointwise local-energy queries over the configurations visited by the sampler. 
By the zero-variance principle, $E_{\mathrm{loc}}(\sigma)$ becomes constant over $\sigma$ if and only if $\psi_\theta$ is an eigenstate of $H$, so local-energy evaluations serve as an efficient verifier of a candidate eigenstate. 
However, they do not reveal where such a state lies within the high-dimensional variational manifold where ground state lived in. 
NQS optimization therefore can naturally be viewed as a search problem equipped with a local verifier but guided only by sampled queries of the objective. 
This perspective explains how subspace restriction can arise during training, if an important configuration $\sigma_0$ is assigned an anomalously small probability, it is rarely sampled, and the information it carries is correspondingly suppressed in the local-energy estimates.
The optimization can then effectively collapse onto the currently explored subspace, making the missing components of the true ground-state wave function increasingly difficult to recover.

\begin{figure}[t]
  \centering
  \includegraphics[width=.8\linewidth]{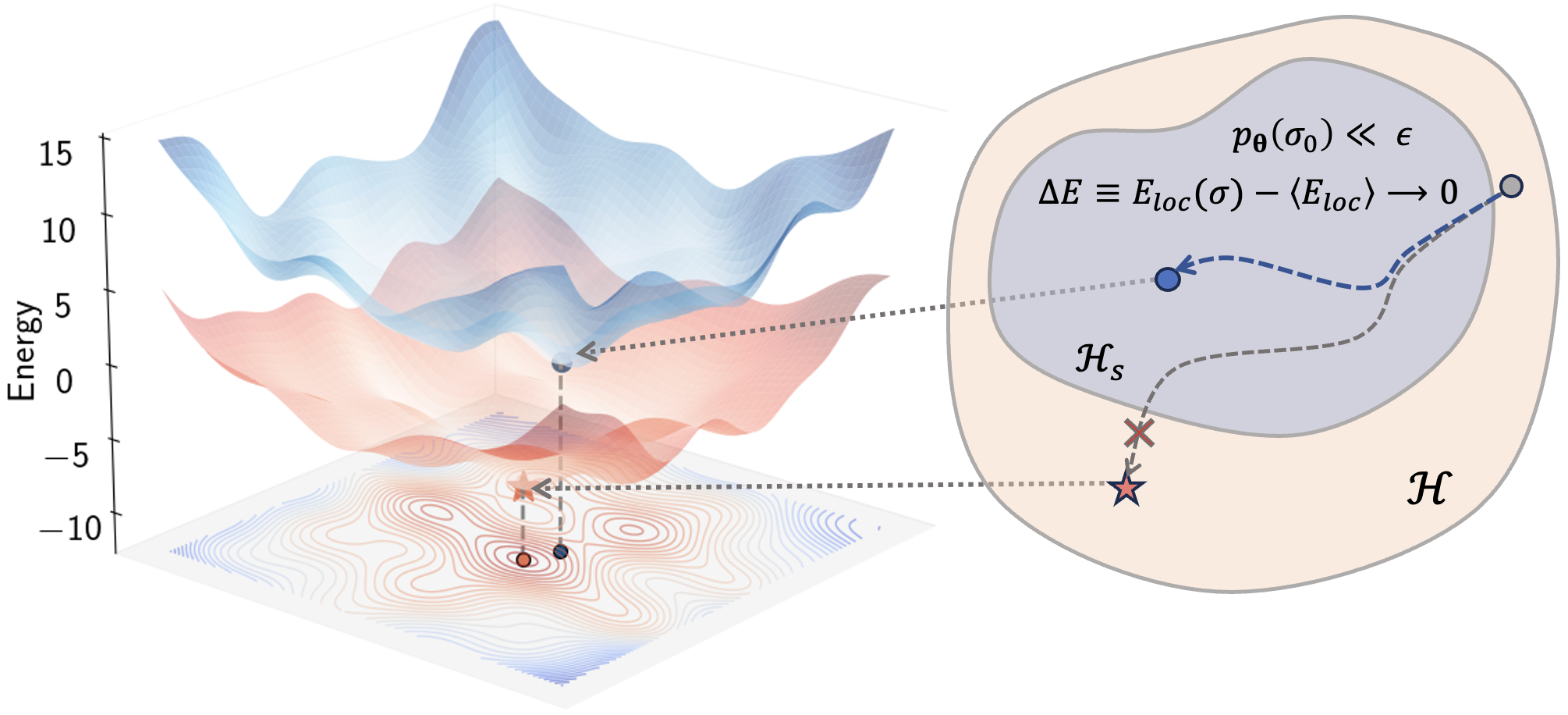}
  \caption{The illustration of subspace trapping in NQS optimization. (Left) Loss landscape in the variational parameter space. (Right) The red and blue regions denote the full Hilbert space $\mathcal{H}$ and a restricted subspace $\mathcal{H}_S$, respectively. The optimization prematurely converges within $\mathcal{H}_S$ (blue curve) when the network underestimates the probability of a critical configuration, $p_{\bm{\theta}}(\sigma_0)\ll \epsilon$, where $\epsilon=1/\Delta t N_s$. Sampling is then restricted to $\mathcal{H}_S$, so the apparent convergence criterion $\Delta E \to 0$ is satisfied even though the energy remains above the true ground-state energy.}
  \label{fig:landscape}
\end{figure}

\begin{figure}[t]
    \centering
        \includegraphics[width=\linewidth]{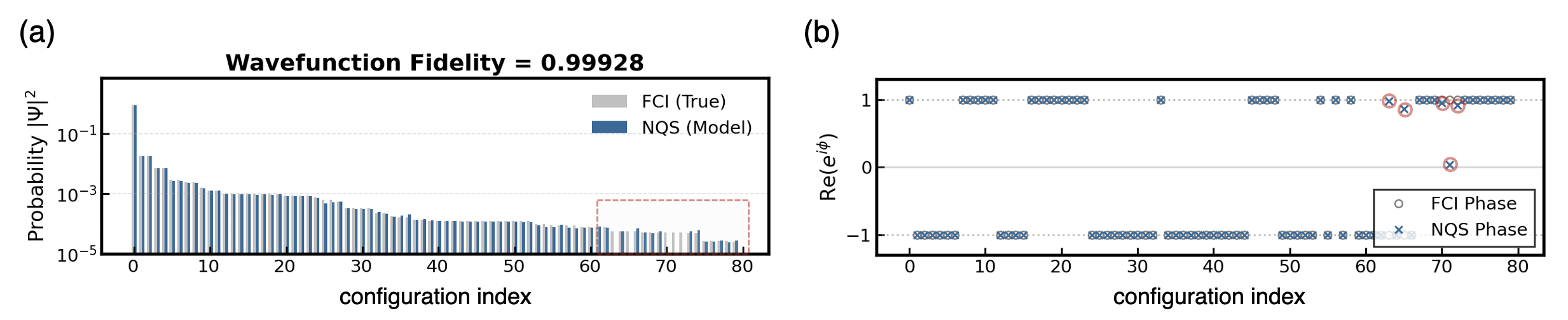}
    \caption{Subspace trapping in NQS optimization.
    Comparison between the exact ground state (Full Configuration Interaction, FCI) and the optimized NQS state in the computational basis, where configurations are sorted by the magnitude of the amplitudes. (a) Probability distribution $|\psi|^2$ of the ground state and the NQS state. The red dashed box highlights the low-weight but non-negligible configurations. (b) Phase comparison between the FCI and NQS states. The red circles mark unsampled or suppressed configurations for which the phase of NQS deviates from the phase of ground state.}
    \label{fig:wf_comparison}
\end{figure}
To access whether the mechanism predicted by Theorem~\ref{thm:Hsubspace} accounts for the local minima observed in NQS optimization, we consider the \ce{N2} molecule as a representative case, with the NQS built by a complex recurrent neural network (cRNN). In Fig.~\ref{fig:wf_comparison}, it illustrates how an NQS can approximate to ground state with high fidelity while still failing to reach chemical accuracy. In Fig.~\ref{fig:wf_comparison}(a), although the NQS reproduces the dominant configurations and achieves a high fidelity, several low-weight but non-negligible configurations are strongly suppressed or missing in the learned wave function. Moreover, as shown in Fig.~\ref{fig:wf_comparison}(b), the suppressed configurations exhibit incorrect phases relative to the exact ground state, which further prevents the NQS from reaching chemical accuracy.
Furthermore, we examined whether the trapped NQS $|\psi_{\bm{\theta}}\rangle$ is optimized within the subspace $\mathcal{H}_{S}$. Given Hamiltonian $H$, let $\sigma_0$ denote the unsampled configuration and define the projected effective Hamiltonian as $\hat{H}=PHP$, where $P=\mathbb{I}-|\sigma_0\rangle\langle\sigma_0|$ is the projector onto $\mathcal{H}_{S}$. We find that the variational energy evaluated with the $H$ is extremely close to the ground-state energy of the projected Hamiltonian $\hat{H}$, with $|\langle H\rangle-\langle \hat{H}\rangle|=3.69\times 10^{-8}\,\,\text{Ha}$. This agreement indicates that the NQS is effectively trapped within $\mathcal{H}_{S}$, and provides direct numerical evidence for subspace trapping and explains how a model can appear converged with high fidelity while still failing to reach chemical accuracy, precisely as predicted by Theorem~\ref{thm:Hsubspace}.

Having established that the optimized NQS is effectively confined to a projected subspace, 
we next ask whether this optimization issue can be overcome by increasing the sampling budget, and why such a remedy becomes insufficient at scale.
Increasing the sample size $N_S$ provides a direct way to mitigate sampling-induced mode-exploration problems, because a larger sample budget improves the chance of observing configurations that are rare under the current model distribution but remain physically important for the ground state. In this sense, increasing $N_S$ can enlarge the effective sampled subspace and help recover update directions that are invisible at smaller sampling resolution. However, this strategy does not directly solve the underlying problem. As the system size grows, the Hilbert-space dimension increases exponentially, and the number of relevant configurations may also grow rapidly. Therefore, the sample size required to reliably cover these configurations can become computationally prohibitive, making brute-force sampling an inefficient and ultimately unscalable remedy.

Advanced optimization methods partially is designed to improve trainability, but they also face a related bottleneck. For example, the minSR algorithm~\cite{Chen2024} improves the scalability of stochastic reconfiguration by exploiting the fact that the sampled quantum geometric tensor has rank at most $N_S$. Instead of solving the natural-gradient equation in the full $N_p$ dimensional parameter space, minSR reformulates the update in the $N_S$ dimensional sample subspace, reducing the per-step cost from $O(N_SN_p^2+N_p^3)$ to $O(N_S^3+N_S^2N_p)$, which makes highly expressive NQS models with large $N_p$ more tractable when $N_S$ remains moderate. Nevertheless, the update directions available to minSR are still restricted to the span of the sampled log-derivatives. If escaping a sampling-induced local minimum requires directions associated with configurations outside the current sample set, then one must increase $N_S$. In precisely this regime, the $O(N_S^3)$ cost becomes the dominant bottleneck, and minSR gradually loses the computational advantage.

Therefore, existing approaches face a fundamental tension between expressivity and sampling coverage. High-accuracy NQS optimization requires a large number of parameters $N_p$ to represent complex many-body states, but it also requires a sufficiently large sample size $N_S$ to cover the relevant configurations and span the critical update directions needed to escape sampling-induced local minima. Conventional SR becomes impractical as $N_p$ grows, while minSR removes the leading $N_p$ bottleneck only by shifting the difficulty to the sampling side. When the required $N_S$ becomes large, its superlinear scaling $O(N_S^3+N_S^2N_p)$ again becomes unaffordable. These limitations indicate that simply increasing the sample budget or relying on sample-subspace optimization is insufficient, and point to the need for a more principled and scalable mechanism that can recover missing physically relevant modes without requiring prohibitively large sampling overhead.

\subsection{AGD-enhanced NQS}
\begin{figure*}[t]
    \centering
        \includegraphics[width=\linewidth]{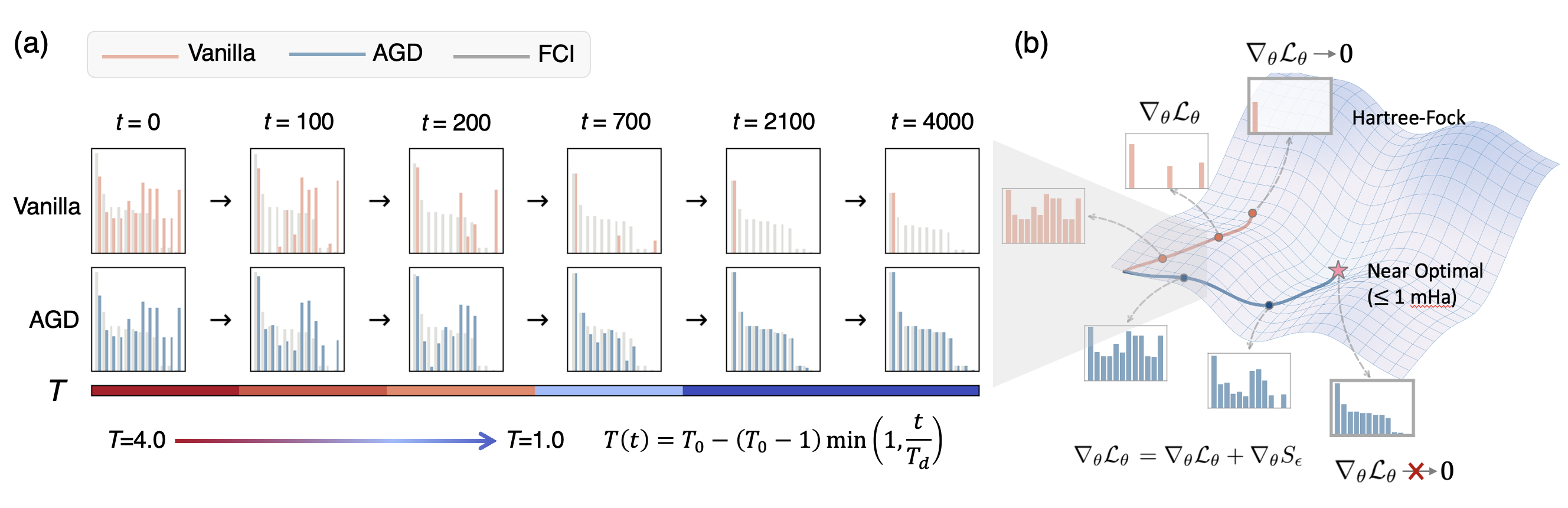}
    \caption{The illustration of an AGD-enhanced NQS for a molecular Hamiltonian as an example, showing that vanilla gradient descent (GD) is prone to be trapped in Hartree-Fock local minima, whereas AGD-enhanced approach enables more effective exploration of the energy landscape. The two colored curves depict the optimization trajectories, the yellow curve corresponds to the NQS optimized with GD, while the blue curve represents the NQS optimized using AGD. The histograms along each trajectory show the corresponding sampled distributions during optimization. Along the GD (red) path, it indicates that GD hard to handle the explore-exploit tradeoff and trends to collapse toward Hartree-Fock like regions with vanishing gradients. In contrast, AGD samples from a tempered distribution controlled by an annealing parameter $T$, which is gradually reduced from $T=4.0$ to $T=1.0$ according to the schedule of $T(t) = T_0 - (T_0 - 1.0) \cdot \min\left(1.0, t/{T_{\mathrm{d}}}\right)$. Along the AGD (blue) trajectory, AGD forces the model to explore and accept more configurations, which leads to a sustained effective descent toward the state with near-optimal energy ($\leq1$ mHa).
}
    \label{fig:framework}
\end{figure*}

To mitigate the computational trade-off between the parameter space dimension $N_p$ and the sample size $N_s$, and to suppress the occurrence of mode collapse, 
we propose an NQS optimization scheme based on annealed enhanced gradient descent (AGD). AGD is motivated by the need to guide the neural network toward a broader set of physically important configurations during the optimization. This prevents the learned probability distribution from prematurely concentrating on a narrow subset of configurations, a behavior that can drive the model toward a spurious ground state. 
By smoothing such probability distributions, AGD enhances mode exploration and helps maintain access to relevant regions of the Hilbert space.
It lies in modifying the optimization dynamics through the introduction of an annealing parameter $T$, which effectively reshapes the gradient landscape during training.
Concretely, the gradients are not evaluated according to the original distribution $p_{\bm{\theta}}$ but over a tempered distribution $q_{\bm{\theta}}$.
\begin{equation}
q_{\bm{\theta}}(\sigma)=\frac{|\psi_{{\bm{\theta}}}(\sigma)|^{\frac{2}{T}}}{\sum_{\sigma'}|\psi_{{\bm{\theta}}}(\sigma')|^{\frac{2}{T}}},
\end{equation}
where $T$ serves as a tuneable temperature parameter during training to encourage exploration of previously undersampled regions. When $T>1$, $q_{\bm{\theta}}$ provides a flattened version of the original distribution, thereby enhancing the visibility of low-probability configurations and weakening the over-representation of dominant configurations in the gradient estimate.
In practice, the annealing parameter $T$ follows a linear decay schedule, starting from an initial temperature $T_{0}$ and decreasing to $T(t)=1.0$ over $T_{d}$ optimization steps,
\begin{equation}
    T(t) = T_0 - (T_0 - 1.0) \cdot \min\left(1.0, \frac{t}{T_{\mathrm{d}}}\right),
\end{equation}
where $t$ denotes the current optimization step. Thus, the annealed probability distribution $q_{\bm{\theta}}$ at $t$-step are computed via
\begin{equation}
q_{\bm{\theta}}(t, \sigma_i) = \frac{\exp\left[\frac{1}{T(t)} \log p_{\bm{\theta}}(\sigma_i)\right]}{\sum_{j=1}^{N_S} \exp\left[\frac{1}{T(t)} \log p_{\bm{\theta}}(\sigma_j)\right]}.
\end{equation}
which effectively implements the $T$-modified distribution $q_{\bm{\theta}}$ through numerically stable log-softmax computation. 
The complete AGD optimization procedure is summarized in Algorithm~\ref{alg:agd} of Appendix~\ref{Appendix:Supplementary Information}.
Furthermore, we now consider the first-order Taylor expansion of new annealed gradient $\hat{\grads}$ over $q_{\bm{\theta}}$ with annealing factor $T=1+\epsilon, \epsilon\rightarrow 0$. It can also be considered as a modified gradient $\hat{\grads}$ written as
\begin{equation}
    \hat{\grads} \leftarrow \grads + \partial _{{\bm{\theta}}}S_{\epsilon}({\bm{\theta}}),
\end{equation}
where $\partial_{{\bm{\theta}}}S_{\epsilon}({\bm{\theta}})$ is the first-order term,
\begin{equation}
\partial_{{\bm{\theta}}}S_{\epsilon}({\bm{\theta}}) = 2\operatorname{Re}\Big[\epsilon \sum_\sigma w(p_{\bm{\theta}}(\sigma))\Delta E(\sigma)\partial_{\bm{\theta}}\psi^*_{\bm{\theta}}(\sigma)\Big],
    \label{equ:analogy_entropy}
\end{equation}
where $w(p_{\bm{\theta}})= -p_{\bm{\theta}}\log p_{\bm{\theta}}$ and $\partial_{{\bm{\theta}}}S_{\epsilon}({\bm{\theta}})$ serves to mitigate mode collapse.  
This property arises because, in the regime where probabilities are small, the gradient update from $\partial_{{\bm{\theta}}}S_{\epsilon}({\bm{\theta}})$ imparts a logarithmic enhancement (scaling with $\log p$) to these components. When considering the $n$-th order approximation, this yields an enhancement scaling as $(\log p)^n$ relative to the original gradient $\partial_\theta \mathcal{L}(\theta)$. A more detailed analysis is provided in Appendix~\ref{sec:gradient}.

Fig.~\ref{fig:framework}(a) shows AGD-enhanced NQS for a molecular Hamiltonian. Under standard gradient descent, the vanilla NQS collapses toward the Hartree-Fock region and converges prematurely. In contrast, AGD first encourages exploration over a broader set of configurations and then anneals toward accurate energy minimization, enabling the variational search to escape subspace trapping and approach the true ground-state structure.

Moreover, beyond its motivation, AGD offers several practical advantages: it requires no modification to the neural-network architecture, introduces negligible computational overhead, and maintains a per-step complexity essentially equivalent to that of standard gradient descent. Moreover, since AGD operates directly at the sampling and gradient-estimation level, it integrates naturally with large-scale autoregressive sampling, which can further enhance sampling coverage and variational accuracy. As the temperature is gradually annealed, the optimization transitions smoothly from broad configuration-space exploration to fine-grained energy minimization, progressively steering the network toward the true ground state.

\subsection{Molecular systems experiments}
\begin{figure}[t]
  \centering
  \includegraphics[width=1\linewidth]{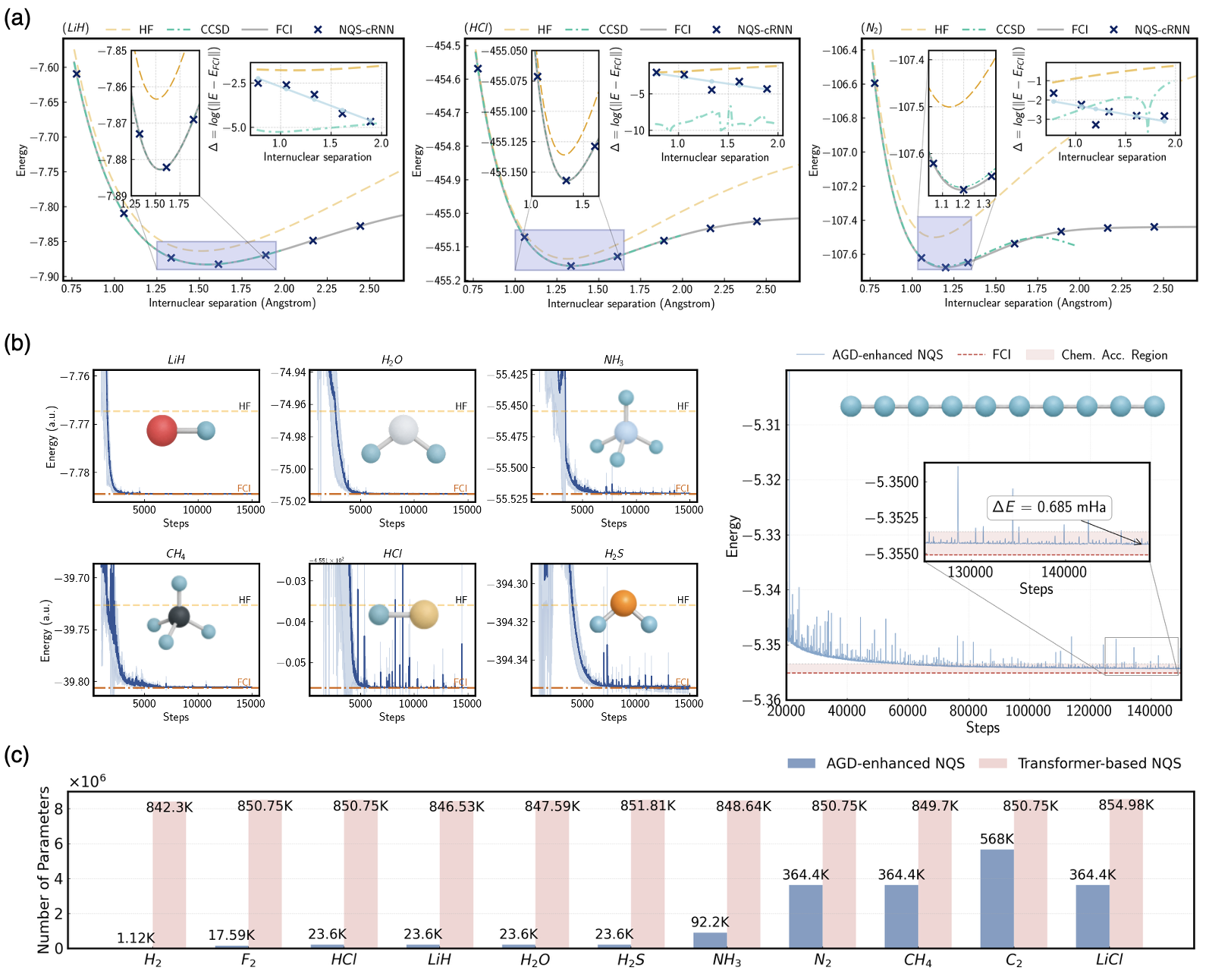}
  \caption{Performance of AGD-enhanced NQS on molecular systems.
  (a) Potential energy surfaces for \ce{LiH}, \ce{HCl}, and \ce{N2} as a function of internuclear separation. The AGD-enhanced cRNN-NQS results are compared with Hartree-Fock (HF), coupled-cluster singles and doubles (CCSD), and full configuration interaction (FCI) references. The zoomed insets highlight the equilibrium region, while the error insets show $\Delta=\log(|E-E_{\mathrm{FCI}}|)$, demonstrating that the AGD-enhanced NQS closely follows the FCI energy across different molecular geometries.
  (b) Training dynamics of AGD-enhanced NQS for molecular ground-state search. The left panels show the average convergence curves for \ce{LiH}, \ce{H2O}, \ce{NH3}, \ce{CH4}, \ce{HCl} and \ce{H2S}, with shaded area presents the spread of results across multiple runs. The right panel shows the optimization trajectory for the larger \ce{H10} chain, with the shaded region denoting chemical accuracy. The inset magnifies the late-stage convergence and reports a final energy error of $\Delta E=0.685$ mHa.
  (c) Comparison of the number of variational parameters used by AGD-enhanced NQS and a transformer-based NQS across different molecular systems. AGD enables compact cRNN-based models to achieve chemical accuracy with substantially fewer parameters. Parameter counts for the transformer-based NQS are taken from Ref.~\cite{shang2025solving}.
}
  \label{fig:mol_results}
\end{figure}

In this section, we benchmark AGD-enhanced NQS, parametrized by a cRNN ansatz, for molecular ground-state energy estimation. Molecular systems provide a stringent benchmark since NQS optimization can easily be trapped near Hartree-Fock like local minima, resulting in premature convergence above chemical accuracy. We therefore evaluate whether AGD helps the NQS escape these well-known local minima and compare its performance with standard quantum chemistry references, including Hartree-Fock (HF), coupled-cluster singles and doubles (CCSD)\cite{Bartlett2007coupled}, and full configuration interaction (FCI) energies.

We first evaluate the proposed method on potential energy surface calculations for \ce{LiH}, \ce{HCl} and \ce{N2} as shown in Fig.~\ref{fig:mol_results}(a).
Across the scanned internuclear separations, the optimized energies of NQS-cRNN reach chemical accuracy across the entire range of bond lengths considered, including the equilibrium region highlighted by the shaded boxes. In contrast, the Hartree-Fock curve shows a clear systematic deviation, especially away from the equilibrium geometry, while CCSD improves the accuracy but still leaves visible residual errors in several bond-length regimes. The logarithmic error insets further show that the AGD-enhanced NQS maintains a small deviation from FCI over a broad range of geometries, indicating that the learned wave function is not restricted to a single equilibrium configuration but remains accurate along the molecular dissociation coordinate.

In Fig.~\ref{fig:mol_results}(b), we next examine the optimization dynamics for a set of representative molecules, including LiH, \ce{H2O}, \ce{NH3}, \ce{CH4}, \ce{HCl} and \ce{H2S}. For all tested systems, the AGD-enhanced NQS rapidly decreases the variational energy from the Hartree-Fock level toward the FCI reference and enters the chemical-accuracy region within a moderate number of training steps. The convergence curves show that AGD is able to stabilize the training process even for molecules with different bonding structures and numbers of electrons. Although transient fluctuations appear during optimization, especially for the more challenging molecules, the final energies remain close to the FCI baseline, which suggests that the AGD strategy improves the exploration of the variational landscape while still allowing accurate energy minimization at later stages. The complementary results are shown in Tab.~\ref{tab:nqs_crnn_molecules}.
The right panel of Fig.~\ref{fig:mol_results}(b) further tests the method on the larger \ce{H10} chain, which presents a more demanding optimization problem due to the enlarged Hilbert space and stronger correlation effects. After a long training trajectory, the energy converges into the shaded chemical-accuracy window around the FCI value. The inset shows that the final energy error reaches $\Delta E=0.685$ mHa, below the standard chemical-accuracy threshold. This result demonstrates that the AGD-enhanced NQS can remain effective beyond small molecular benchmarks and can achieve chemically accurate energies in a more extended molecular system.

Finally, we compares the number of trainable parameters required by the AGD-enhanced NQS and a vanilla transformer-based NQS across different molecular Hamiltonians shown in Fig.~\ref{fig:mol_results}(c). The vanilla transformer based model utilizes large parameter budget, whereas the AGD-enhanced NQS requires substantially fewer parameters for all tested molecules. The reduction is especially pronounced for small and medium-sized systems such as \ce{H2}, \ce{F2}, \ce{HCl}, \ce{LiH}, \ce{H2O}, \ce{H2S} and \ce{NH3}, where the AGD-enhanced model achieves accurate results with orders-of-magnitude fewer parameters. Even for more complex systems such as \ce{N2}, \ce{CH4}, \ce{C2} and \ce{LiCl}, the parameter count remains below that of the transformer-based NQS. Together, these results indicate that AGD improves both the optimization efficiency and parameter efficiency of neural quantum-state calculations, enabling chemically accurate molecular ground-state estimation with a compact variational ansatz.

\begin{table*}[t]
\centering
\renewcommand{\arraystretch}{1.18}
\setlength{\tabcolsep}{10pt}
\setlength{\arrayrulewidth}{1pt}
\begin{tabular*}{0.9\textwidth}{@{\extracolsep{\fill}}lccccccc}
\hline
Molecule 
& $N_{\mathrm{o}}$ 
& $N_{\mathrm{e}}$ 
& HF Energy 
& CCSD 
& FCI 
& \multicolumn{2}{c}{NQS-cRNN} \\
\cline{7-8}
& & & & & & Ours & Settings \\
\hline
\ce{H2} & 4 & 2 & -0.9109 & -0.9981 & \underline{-0.9981} & \bf{-0.9981} & \texttt{Hiddens=10, $T_0$=4.0} \\
\ce{F2} & 20 & 18 & -195.6380 & \bf{-195.6611} & \underline{-195.6611} & -195.6610 & \texttt{Hiddens=43, $T_0$=4.0} \\
\ce{HCl} & 20 & 18 & -455.1360 & -455.1562 & \underline{-455.1562} & \bf{-455.1562} & \texttt{Hiddens=50, $T_0$=4.0} \\
\ce{LiH} & 12 & 4 & -7.7674 & -7.7845 & \underline{-7.7845} & \bf{-7.7845} & \texttt{Hiddens=50, $T_0$=4.0} \\
\ce{H2O} & 14 & 10 & -74.9644 & -75.0154 & \underline{-75.0155} & \bf{-75.0155} & \texttt{Hiddens=50, $T_0$=4.0} \\
\ce{H2S} & 22 & 18 & -394.3114 & \bf{-394.3546} & \underline{-394.3546} & -394.3541 & \texttt{Hiddens=50, $T_0$=4.0} \\
\ce{NH3} & 16 & 10 & -55.4548 & \bf{-55.5209} & \underline{-55.5211} & -55.5206 & \texttt{Hiddens=100, $T_0$=4.0} \\
\ce{N2} & 20 & 14 & -107.4990 & -107.6560 & \underline{-107.6602} & \bf{-107.6596} & \texttt{Hiddens=200, $T_0$=4.0} \\
\ce{CH4} & 18 & 10 & -39.7266 & -39.8060 & \underline{-39.8063} & \bf{-39.8062} & \texttt{Hiddens=200, $T_0$=4.0} \\
\ce{C2} & 20 & 12 & -74.4209 & -74.6745 & \underline{-74.6908} & \bf{-74.6901} & \texttt{Hiddens=250, $T_0$=4.0} \\
\ce{LiCl} & 28 & 20 & -460.8273 & -460.8476 & \underline{-460.8496} & \bf{-460.8491} & \texttt{Hiddens=200, $T_0$=5.0} \\
\hline
\end{tabular*}
\caption{Comparison of molecular ground-state energies obtained by HF, CCSD, FCI and NQS-cRNN.}
\label{tab:nqs_crnn_molecules}
\end{table*}

\begin{figure}[t]
  \centering
  \includegraphics[width=\linewidth]{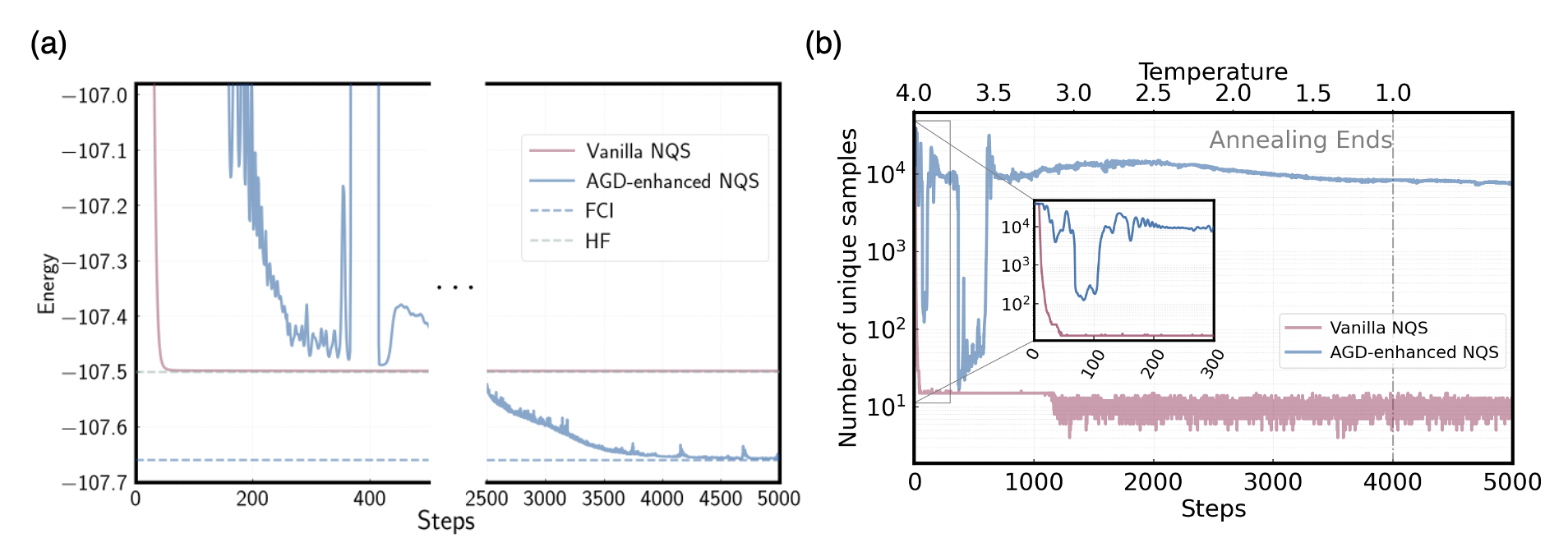}
  \caption{Annealing from a higher initial temperature enables convergence to the ground state for $N_2$. 
  (a) Optimization curve of vanilla NQS and AGD-enhanced NQS.
  (b) Number of unique sampled configurations during optimization for vanilla NQS and AGD-enhanced NQS, together with the annealing temperature schedule shown on the upper axis. }
  \label{fig:unique_samples}
\end{figure}

Moreover, we examine whether AGD-enhanced NQS avoids the sample collapse. As shown in Fig.~\ref{fig:unique_samples}(a), the vanilla NQS rapidly lowers the variational energy during the early stage of optimization, but soon becomes trapped near the Hartree-Fock state, which is accompanied by a sharp collapse of the sampled configuration space, with the number of unique samples remaining only at $O(10)$ in Fig.~\ref{fig:unique_samples}(b). Such a sharp reduction in sampling diversity explains why once important configurations are no longer sampled, the optimizer receives little information about directions that could improve the wave function.
In contrast, the AGD-enhanced NQS exhibits stronger fluctuations at the beginning of training, with high tempreture, which prevents the probability distribution from concentrating too early on a small subset of configurations, reflecting a more exploratory search over active configuration space, with approximately $O(10^4)$ unique samples throughout most of training. 
As the temperature is gradually annealed, the sampling distribution becomes more focused, allowing the model to refine the energy while preserving enough diversity to avoid subspace trapping. 
After the annealing stage ends, the number of unique samples decreases moderately but remains orders of magnitude larger than that of the vanilla NQS. 
Together, these results show that AGD improves NQS training by separating the optimization into an exploration-dominated stage followed by an annealed energy-minimization stage, leading to both better sampling coverage and improved convergence toward the ground-state energy.

\subsection{Frustrated spin systems: J1-J2 model}

Here, we evaluates the AGD-enhanced NQS on frustrated $J_1$-$J_2$ spin models, which provide representative benchmarks for strongly correlated lattice models. 
\begin{equation}
H = J_1\sum_{\langle i,j\rangle}^N \bm{s}_i\bm{s}_j + J_2\sum_{\langle \langle i,j\rangle\rangle}^N \bm{s}_i\bm{s}_j,
\end{equation}
where $\bm{s}_i$ represent spin operator acting on lattice site $i$.
To quantify the accuracy of the proposed method and compared against the other approaches, we the relative error $\epsilon$ and V-score \cite{Wu2024variational},
\begin{align}
    \epsilon &= \frac{\lvert E - E_0 \rvert}{\lvert E_0 \rvert},\\
    V_\text{score} &= \frac{N\,\mathrm{Var}(E)}{E^{2}},
   \label{eq:Vscore}
\end{align}
where $E$ and $\mathrm{Var}(E)$ refers to variational energy and the corresponding variance, $E_0$ denotes the ground-state energy, $N$ is the system size.
We first benchmark AGD-enhanced NQS on the one-dimensional $J_1$-$J_2$ Heisenberg model with $N=16$ sites and periodic boundary conditions, scanning the coupling ratio $J_2/J_1$ from $0$ to $0.8$. As shown in Fig.~\ref{fig:J1J2_combined}(a), the relative error $\epsilon$ closely tracks the V-score and remains small across the entire range, demonstrating that the AGD-enhanced NQS accurately captures the ground state in both weakly and moderately frustrated regimes.
In the shaded strongly frustrated regime, the relative error remains below $\sim 10^{-4}$, showing that the AGD-enhanced NQS remains accurate even where competing interactions render the ground-state sign structure nontrivial and the optimization landscape more challenging.
\begin{figure}[t]
  \centering
  \includegraphics[width=.9\linewidth]{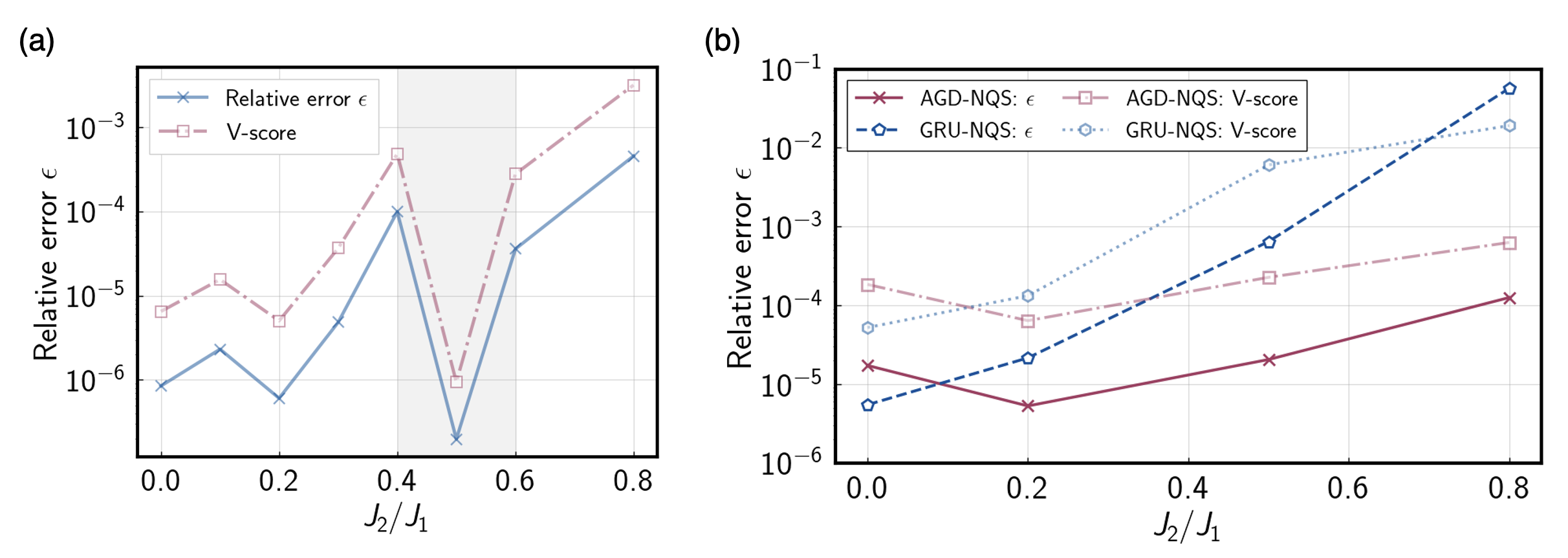}
    \caption{Accuracy of the AGD-enhanced NQS for the $J_1$-$J_2$ model as a function of the coupling ratio $J_2/J_1$. 
    (a) Relative error $\epsilon$ and V-score for the one-dimensional chain with $N = 16$ sites. The shaded region indicates the strongly frustrated regime near the Majumdar-Ghosh point $J_2/J_1 = 0.5$. 
    (b) Relative error $\epsilon$ and V-score for the two-dimensional [$4\times 4$ lattice], comparing AGD-enhanced NQS with GRU-NQS.}
    \label{fig:J1J2_combined}
\end{figure}

We further assess the robustness of AGD-enhanced NQS on the 2D frustrated $J_{1}$-$J_{2}$ Heisenberg model over a broad range of coupling ratios, $J_{2}/J_{1}\in[0,0.8]$. As shown in Fig.~\ref{fig:J1J2_combined}(b), it further highlights the advantage of AGD over the baseline NQS. While both methods attain relatively small errors in the weakly frustrated regime, the baseline deteriorates rapidly as $J_{2}/J_{1}$ increases, reaching a relative error of approximately $10^{-3}$ at $J_{2}/J_{1}=0.5$ and $10^{-2}$ at $J_{2}/J_{1}=0.8$. By contrast, AGD-NQS maintains relative errors of approximately $10^{-5}$ and $10^{-4}$, respectively, together with substantially lower V-scores. These results indicate that the advantage of AGD becomes increasingly pronounced as frustration grows, suggesting that its sampling-aware gradient correction prevents the optimization from concentrating on the subspace and thereby enables a more faithful description of strongly correlated ground states.
Overall, these results demonstrate that AGD-enhanced NQS has promising potential for treating strongly correlated systems beyond molecular Hamiltonians. The method maintains good accuracy in frustrated spin models and provides a meaningful hardness indicator through the V-score.

\section{Discussion}
\label{sec:Discussion}

In this work, we attempt to address a fundamental obstacle in NQS optimization, the tendency of training to converge to spurious local minima long before the true ground-state manifold is reached. 
We first provide theoretical analysis (Theorem~\ref{thm:Hsubspace}) which explain the reason why this occurs.
Because the gradient in NQS optimization is estimated stochastically from sampled configurations, configurations that are statistically rare, but physically important for capturing correlation, can be systematically underweighted during optimization. 
Such sampling bias causes the NQS effectively explore in a subspace which diverges from the actual ground-state manifold, leading it ultimately trapped into the local minima in such subspace. 

To mitigate this problem, we introduce a simple and effective approach to enhanced NQS with annealed optimization strategy designed to prevent configuration collapse. The central idea is to incorporate a temperature-dependent annealing factor into the gradient update, thereby reweighting the contribution of sampled configurations during training. This controlled reweighting encourages the optimizer to continue exploring configurations that may carry essential correlation or sign information but would otherwise be rapidly eliminated from the sampling process.

We further validate our approach through numerical experiments on a range of molecular and strongly correlated many-body systems, where such trapping phenomena can naturally arise during NQS optimization.
With all tested cases, AGD-enhanced NQS both attains chemical accuracy. They also provide direct evidence of the necessity of AGD, without its corrective influence, the optimization is susceptible to becoming trapped in the Hartree-Fock state, failing to capture the essential correlation. Remarkably, we noticed that the enhancement of AGD achieved without increasing the computational overhead, as the it possesses the same computational complexity as standard gradient descent, which allows for the exploitation of large-sample autoregressive sampling to achieve higher accuracy.

More broadly, our findings indicate that improving the optimization dynamics of NQS requires not only more expressive neural architectures, but also training rules that explicitly account for the statistical structure of the sampled Hilbert space. Looking forward, integrating sampling-aware optimization with advanced neural architectures and adaptive sampling schemes may further enhance the robustness and scalability of NQS methods for tackling increasingly complex quantum systems.

\section{Methods}
\label{sec:Methods}

\subsection{Variational Monte Carlo Framework}
We consider the general quantum many-body problem defined by a Hamiltonian $\hat{H}$ acting on a Hilbert space $\mathcal{H}$ with dimension $D = |\mathcal{H}|$. For a lattice system with $N_s$ sites and local Hilbert space dimension $d$, we have $D = d^{N_s}$, which grows exponentially with system size. The ground state $|\Psi_0\rangle$ and its energy $E_0 = \langle\Psi_0|\hat{H}|\Psi_0\rangle/\langle\Psi_0|\Psi_0\rangle$ are determined by the eigenvalue equation $\hat{H}|\Psi_0\rangle = E_0|\Psi_0\rangle$.

According to the variational principle, the ground state energy can be obtained by minimizing the energy functional
\begin{equation}
    E[\psi] = \frac{\langle\psi|\hat{H}|\psi\rangle}{\langle\psi|\psi\rangle}
    \geq E_0,
    \label{eq:variational_principle}
\end{equation}
over a trial wave function manifold. The equality holds if and only if $|\psi\rangle = |\Psi_0\rangle$ (up to a global phase factor). 
In NQS approach, the trial wave function is parametrized by a neural network with real parameters $\bm{\theta} \in \mathbb{R}^{N_p}$:
\begin{equation}
    \psi_{\bm{\theta}}(\sigma) = \langle\sigma|\psi_{\bm{\theta}}\rangle,
\end{equation}
where $\{|\sigma\rangle\}$ denotes a complete basis of the Hilbert space (e.g., the computational basis for spin systems). The wave function can be decomposed into amplitude and phase as \begin{equation}
    \psi_{\bm{\theta}}(\sigma) = \sqrt{p_{\bm{\theta}}(\sigma)}\, e^{i\phi_{\bm{\theta}}(\sigma)},
\end{equation}
where $p_{\bm{\theta}}(\sigma) = |\psi_{\bm{\theta}}(\sigma)|^2 / Z_{\bm{\theta}}$ is the normalized probability distribution ($Z_{\bm{\theta}} = \sum_\sigma |\psi_{\bm{\theta}}(\sigma)|^2$). The neural network predicts the probability $p_{\bm{\theta}}(\sigma)$ for configuration $\sigma$, while the phase $\phi_{\bm{\theta}}(\sigma)$ is generated by an auxiliary network branch.

The variational energy Eq.~\ref{eq:variational_principle} can be rewritten in terms of the local energy
\begin{equation}
    E_{\mathrm{loc}}(\sigma) = \frac{\langle\sigma|\hat{H}|\psi_{\bm{\theta}}\rangle}{\langle\sigma|\psi_{\bm{\theta}}\rangle} = \sum_{\sigma'} H_{\sigma\sigma'} \frac{\psi_{\bm{\theta}}(\sigma')}{\psi_{\bm{\theta}}(\sigma)},
    \label{eq:local_energy}
\end{equation}
as
\begin{equation}
    E(\bm{\theta}) = \sum_{\sigma} p_{\bm{\theta}}(\sigma) E_{\mathrm{loc}}(\sigma) = \mathbb{E}_{p_{\bm{\theta}}}[E_{\mathrm{loc}}].
    \label{eq:energy_expectation}
\end{equation}
Since dimension of the Hilbert space grows exponentially with system size, rendering the exact evaluation of Eq.~\ref{eq:energy_expectation} computationally intractable. The local energy is instead evaluated as the expectation value $\mathbb{E}_{p_{\bm{\theta}}}[E_{\mathrm{loc}}]$, which is approximated via stochastic sampling. Specifically, $N_S$ configurations $\{\sigma^{(k)}\}_{k=1}^{N_S}$ are drawn independently from the neural network's probability distribution $p_{\bm{\theta}}(\sigma)$, yielding the estimator
\begin{equation}
    E(\bm{\theta}) = \frac{1}{N_S} \sum_{k=1}^{N_S} E_{\mathrm{loc}}(\sigma^{(k)}).
    \label{eq:energy_estimator}
\end{equation}
Differentiating the energy estimator in Eq.~\ref{eq:energy_estimator} with respect to the network parameters $\bm{\theta}$ yields the parameter gradient expressed in Eq.~\ref{eq:vmc_gradient}. The parameters are then updated via gradient descent in stand NQS based VMC.

The optimization proceeds iteratively: samples are generated from the current distribution $p_{\bm{\theta}}(\sigma)$, gradients are estimated, and parameters are updated until convergence to a stationary point is achieved. The quality of the final solution depends critically on (i) the expressivity of the ansatz (i.e., whether the parameterized manifold contains a good approximation to $|\Psi_0\rangle$), and (ii) the ability of the sampling-based optimization to navigate the high-dimensional parameter landscape. As discussed in the main text, the latter is often hampered by  sampling-based optimization bottlenecks that we address with the AGD method.

\subsection{cRNN Architecture and Ancestral Sampling}

The trial wave function in our framework is parametrized by a complex-valued recurrent neural network (cRNN). The network architecture closely follows the RNN wave function ansatz introduced by Hibat-Allah \emph{et al.}~\cite{HibatAllah2020}, while the sampling procedure--adapted to molecular systems with a fixed particle number--builds on the autoregressive scheme of Barrett \emph{et al.}~\cite{Barrett2022}. Together, these ingredients endow the ansatz with both expressive power and exact, symmetry-respecting sampling, which are essential for the AGD optimization developed in this work.

\subsubsection{Autoregressive wave function parametrization}

For a configuration $\sigma = (\sigma_1, \sigma_2, \dots, \sigma_N)$ in the computational basis, where $N$ is the number of sites (or spin-orbitals in molecular systems) and $\sigma_i \in \{0, 1, \dots, d-1\}$ denotes the local state at site $i$, the cRNN factorizes the joint probability into a product of conditionals:
\begin{equation}
    p_{\bm{\theta}}(\sigma) = \prod_{i=1}^{N} p_{\bm{\theta}}(\sigma_i \mid \sigma_1, \dots, \sigma_{i-1}).
    \label{eq:autoregressive_factorization}
\end{equation}
Each conditional is produced by a softmax over a learned value function $\mathcal{V}_{\bm{\theta}}$ with complex parameters $\bm{\theta}$,
\begin{equation}
    p_{\bm{\theta}}(\sigma_i \mid \sigma_{<i}) = \frac{\exp\left[-\mathcal{V}_{\bm{\theta}}(\sigma_i, \mathbf{h}_{i-1})\right]}{\sum_{s=0}^{d-1} \exp\left[-\mathcal{V}_{\bm{\theta}}(s, \mathbf{h}_{i-1})\right]},
    \label{eq:conditional_prob}
\end{equation}
where $\mathbf{h}_{i-1}$ is the hidden state encoding the history $\sigma_{<i}$. The hidden state is propagated by a gated recurrent unit (GRU),
\begin{equation}
    \mathbf{h}_i = \mathrm{GRU}_{\bm{\theta}}(\sigma_i, \mathbf{h}_{i-1}),
    \label{eq:gru_update}
\end{equation}
in direct analogy with the construction of Ref.~\cite{HibatAllah2020}.

The phase of the wave function is generated by an independent linear head applied to the same sequence of hidden states and accumulated across sites,
\begin{equation}
    \phi_{\bm{\theta}}(\sigma) = \sum_{i=1}^{N} \mathrm{softsign}\left[\mathbf{W}_{\mathrm{phase}} \cdot \mathbf{h}_{i-1} + b_{\mathrm{phase}}\right],
    \label{eq:phase_network}
\end{equation}
with the softsign activation $\mathrm{softsign}(x) = \pi x / (1 + |x|)$ confining each contribution to $(-\pi, \pi)$. Combining amplitude and phase yields the full complex wave function
\begin{equation}
    \psi_{\bm{\theta}}(\sigma) = \sqrt{p_{\bm{\theta}}(\sigma)} \, e^{i\phi_{\bm{\theta}}(\sigma)}.
    \label{eq:psi_amplitude_phase}
\end{equation}

\subsubsection{Sampling with physical constraints}

A key benefit of the autoregressive form Eq.~\ref{eq:autoregressive_factorization} is that it admits exact ancestral sampling: configurations are drawn site by site from the conditional distributions of Eq.~\ref{eq:conditional_prob}, producing independent and identically distributed (i.i.d.) samples and thereby avoiding the autocorrelation overhead inherent to Markov-chain Monte Carlo schemes.

For molecular systems with a fixed particle number $N_{\mathrm{part}}$, naive ancestral sampling does not respect the global $U(1)$ symmetry of the Hamiltonian. To enforce particle-number conservation exactly, we employ the constrained-sampling scheme of Hibat-Allah \textit{et al.}~\cite{HibatAllah2020}, wherein the conditional distributions are masked on the fly such that only configurations within the physical sector are generated. Specifically, at step $t$, given the running occupation $n_t = \sum_{s=1}^{t}\sigma_s$, a candidate occupation $\sigma_t = s$ is rejected whenever either (i) $n_t \geq N_{\mathrm{part}}$, in which case the particle budget is exhausted, or (ii) $n_t + (N-t) < N_{\mathrm{part}}$, in which case the remaining sites are insufficient to accommodate the missing particles. The surviving probabilities are subsequently renormalized to ensure that each conditional remains a properly normalized distribution. This procedure is exact in the sense that every sampled configuration carries the prescribed particle number, while the independent and identically distributed nature of ancestral sampling is fully preserved. To efficiently exploit modern hardware, we further adopt the autoregressive sampling algorithm of Barrett \textit{et al.}~\cite{Barrett2022}, which is particularly well suited for large-batch generation.

\begin{acknowledgments}
This work is supported by 
Quantum Science and Technology-National Science and Technology Major Project (2023ZD0300200), 
the National Natural Science Foundation of China Grant (No.~12361161602), 
NSAF (Grant No.~U2330201), 
Beijing Natural Science Foundation Z250004, 
Beijing Science and Technology Planning Project (Grant No.~Z25110100810000). The numerical experiments of this work are supported by the High-performance Computing Platform of the Institute of High Energy Physics (Beijing) and Peking University.
\end{acknowledgments}

\bibliographystyle{unsrt}
\bibliography{mybib.bib}

\newpage
\appendix
\section{Related works on solving training issues of NQS}
Although neural quantum states provide expressive variational representations, their optimization remains intrinsically vulnerable to sampling-induced failures. 
In sampling-based training, the variational energy is evaluated only through configurations visited by the current model distribution. 
As a result, regions of configuration space that are rarely sampled receive little or no gradient, even if they contain physically important contributions. 
This creates a problematic feedback loop: configurations that are sampled frequently and appear to lower the energy are preferentially reinforced, whereas configurations with an incorrectly learned sign structure are assigned progressively lower probabilities and may eventually become undersampled. 
Once the probabilities of these configurations are severely underestimated and fall far below the sampling resolution, they may remain unsampled over many iterations, rendering them effectively invisible to both local-energy estimates and gradient updates. This failure can manifest as mode collapse, exponentially suppressed configurations, and spurious modes.
Mode collapse refers to a degenerative phenomenon where the neural network loses its ability to represent the true quantum state, converging instead into trivial, unphysical, or highly constrained states, exponentially suppressed configurations describe individual configurations whose probability is underestimated by an exponential factor even when the global distribution appears accurate, and spurious modes correspond to artificial high-probability regions created by the variational model in poorly sampled parts of configuration space. 
In all cases, the training dynamics may appear stable according to standard diagnostics, while the learned state only represents a restricted portion of the relevant Hilbert or configuration space.

Several works have clarified these pathologies in closely related variational neural sampling problems. Variational autoregressive networks showed that minimizing a reverse-KL or variational-free-energy objective can suffer from mode collapse. The model may correctly reproduce the relative Boltzmann weights within the modes it covers, while completely missing other modes \cite{wu2019solving}. 
A key implication is that low variance of the energy-like training signal is not a sufficient certificate of global correctness, because the zero-variance condition only constrains the support already covered by the model. 
This observation explains why a variational distribution can appear well optimized while still yielding an incorrect free energy or missing symmetry-related sectors. Temperature annealing and explicit symmetry restoration were proposed as practical remedies, allowing the model to start from a high-temperature distribution with broad support and then follow the target distribution as the temperature is lowered \cite{wu2019solving}. 
Variational neural annealing further systematized this idea by introducing entropy-driven classical annealing and driver-Hamiltonian-based quantum annealing within an autoregressive variational framework, showing that annealed training can substantially outperform direct optimization on several spin-glass and optimization problems \cite{hibat2021variational}.

At the same time, subsequent studies revealed important limitations of annealing and neural sampling. In the Newman-Moore model, variational neural annealing was found to locate low-energy states while still collapsing onto only a subset of degenerate ground-state configurations, indicating that correct energy alone does not guarantee correct representation of the full ground-state manifold \cite{inack2022neural}. 
More broadly, machine-learning assisted Monte Carlo was shown to fail on computationally hard sampling problems such as random graph coloring: variational training tends to collapse onto a small region of phase space, whereas data-driven training can suffer from expressivity and overfitting limitations that destroy global MCMC acceptance \cite{ciarella2023machine}. 
These results suggest that difficult configuration-space geometry can be transferred into the trainability of the neural sampler. In such regimes, the issue is not merely whether the neural network is expressive enough in principle, but whether the training dynamics can maintain access to all relevant modes before they are lost.

Other works have focused on correcting the sampling process itself. The concept of exponentially suppressed configurations was introduced to show that even a well-trained neural sampler can severely underestimate a small number of configurations because the reverse-KL objective penalizes such errors only logarithmically \cite{wu2021unbiased}. 
These rare configurations may have negligible impact on the variational objective, but can cause severe failures in importance sampling or neural global-update schemes by producing large weight fluctuations and long autocorrelation times. 
Neural cluster updates, combined with symmetry operations, were therefore proposed to restore unbiased Monte Carlo sampling and help Markov chains escape such suppressed configurations \cite{wu2021unbiased}. 
In variational Monte Carlo for NQS, spurious modes provide a complementary failure mode: under-sampled regions may develop artificial high-probability density, leading to sudden and persistent energy spikes when walkers enter these regions \cite{zhang2023understanding}. Collective-variable VMC addresses this by adding a physically motivated penalty term that suppresses probability accumulation outside the relevant region of a chosen collective variable, thereby eliminating spurious energy spikes and improving energy accuracy \cite{zhang2023understanding}.

These prior approaches identify and mitigate important aspects of neural-sampling failure, but they act at different levels from annealed gradient descent. Temperature annealing and variational neural annealing modify the training path by introducing thermal entropy or quantum-driving terms. 
Neural cluster updates and symmetry operations modify the sampling dynamics to preserve unbiasedness and reduce autocorrelation, collective-variable VMC introduces an additional diagnostic and penalty based on predefined physical coordinates, and machine-learning-assisted Monte Carlo studies emphasize the intrinsic difficulty of representing and sampling hard multimodal distributions. 
AGD instead targets the gradient estimator used in NQS optimization itself. 
It does not require replacing the sampler, introducing a new physical collective variable, enforcing a particular symmetry, changing the Hamiltonian, or solving an additional curvature problem. 
By applying a temperature-dependent reweighting directly to sampled gradient contributions, AGD counteracts the premature dominance of high-probability configurations and amplifies the influence of low-probability configurations before they become statistically invisible. 
In this sense, AGD is a gradient-level correction to sampling-induced bias, which is designed not merely to improve sample quality after collapse has occurred, but to prevent the optimization trajectory from collapsing into an artificially restricted subspace in the first place.

\section{Mapping molecular ecectronic Hamiltonian to qubit Hamiltonian}
Within the Born-Oppenheimer approximation, the nuclei are clamped at
fixed positions $\{\mathbf{R}_I\}$ and the nonrelativistic electronic
Hamiltonian of a molecule with $N_e$ electrons reads, in atomic units,
\begin{equation}
  \hat{H}_{\mathrm{el}}
  = -\frac{1}{2}\sum_{i=1}^{N_e}\nabla_i^{2}
    -\sum_{i=1}^{N_e}\sum_{I}\frac{Z_I}{\lvert \mathbf{r}_i-\mathbf{R}_I\rvert}
    +\frac{1}{2}\sum_{i\neq j}\frac{1}{\lvert \mathbf{r}_i-\mathbf{r}_j\rvert},
  \label{eq:first_quantized}
\end{equation}
where $Z_I$ denotes the charge of nucleus $I$; the constant
internuclear repulsion is omitted as it merely shifts the spectrum.
To obtain a finite-dimensional representation amenable to quantum
simulation, the electronic degrees of freedom are projected onto a
finite set of $N$ orthonormal spin orbitals
$\{\phi_p(\mathbf{x})\}_{p=1}^{N}$, with
$\mathbf{x}=(\mathbf{r},\sigma)$ a combined spatial-spin coordinate.
These orbitals are typically obtained from a preceding self-consistent
Hartree-Fock calculation in an atom-centered Gaussian basis set
(e.g., STO-3G or cc-pVDZ)~\cite{Helgaker2000}. In this basis,
Eq.~\ref{eq:first_quantized} takes the second-quantized form
\begin{equation}
  \hat{H}
  = \sum_{p,q=1}^{N} h_{pq}\,\hat{a}^{\dagger}_{p}\hat{a}_{q}
  + \frac{1}{2}\sum_{p,q,r,s=1}^{N} h_{pqrs}\,
    \hat{a}^{\dagger}_{p}\hat{a}^{\dagger}_{q}\hat{a}_{r}\hat{a}_{s},
  \label{eq:second_quantized}
\end{equation}
where $\hat{a}^{\dagger}_{p}$ ($\hat{a}_{p}$) creates (annihilates) an
electron in spin orbital $\phi_p$. The Hamiltonian \eqref{eq:second_quantized} acts on the Fock space spanned by occupation-number states $\lvert x_{N-1},\dots,x_{1},x_{0}\rangle$ with $x_{p}\in\{0,1\}$, whose dimension $2^{N}$ motivates the identification of each fermionic mode with one qubit.

Qubit operators are distinguishable and commute on different sites, whereas the fermionic operators in Eq.~\ref{eq:second_quantized} anticommute; a faithful mapping must therefore encode the antisymmetric exchange statistics into nonlocal strings of Pauli operators. The canonical choice is the Jordan-Wigner (JW) transformation~\cite{JordanWigner1928}, which stores the occupation of mode $p$ locally in qubit $p$ and attaches a parity string acting on all preceding qubits,
\begin{equation}
  \hat{a}^{\dagger}_{p}
  \;\longmapsto\;
  \left(\prod_{q<p}\sigma^{z}_{q}\right)\sigma^{-}_{p},
  \qquad
  \hat{a}_{p}
  \;\longmapsto\;
  \left(\prod_{q<p}\sigma^{z}_{q}\right)\sigma^{+}_{p},
  \label{eq:JW}
\end{equation}
with $\sigma^{\pm}_{p}=(\sigma^{x}_{p}\mp i\sigma^{y}_{p})/2$. Direct substitution of Eq.~\ref{eq:JW} into Eq.~\ref{eq:second_quantized} and normal ordering of the resulting spin operators yields a qubit Hamiltonian in the general form of a weighted sum of Pauli strings,
\begin{equation}
  \hat{H}_{Q}
  = \sum_{j} h_{j}\,\hat{P}_{j},
  \qquad
  \hat{P}_{j}
  = \bigotimes_{p=0}^{N-1}\sigma^{\nu_{j,p}}_{p},
  \quad
  \nu_{j,p}\in\{\mathbb{I},x,y,z\},
  \label{eq:pauli_sum}
\end{equation}
where the real coefficients $h_{j}$ are linear combinations of the molecular integrals. Because the two-body term in
Eq.~\ref{eq:second_quantized} carries four orbital indices, the number of Pauli strings scales as $\mathcal{O}(N^{4})$, while the JW strings themselves have weight (number of nontrivial Pauli factors) up to $\mathcal{O}(N)$. The qubit Hamiltonian of Eq.~\ref{eq:pauli_sum} defines the computational basis and the operator whose expectation value is minimized in classical variational approaches based on neural-network quantum states. 

\section{Experiment details}

\label{Appendix:Supplementary Information}
This section provides additional details supporting the results presented in the main text. Our wavefunction ansatz is a complex-valued recurrent neural network~\cite{HibatAllah2020} in which the amplitude and phase of the wavefunction share a common recurrent backbone. The backbone consists of two stacked GRU layers with an equal number of hidden units per layer , which process the qubit (spin) configuration autoregressively, one site at a time. At each site, the hidden state of the final GRU layer is passed to two separate linear output heads: an amplitude head, which produces the conditional probability distribution over the local two-dimensional Hilbert space via a softmax activation, and a phase head, which is  composed  of a single linear layer followed by a softsign activation, which outputs the local contribution to the phase.
For the sampling procedure, we adopt the autoregressive sampling scheme proposed by Barrett \emph{et al.}~\cite{Barrett2022}, which efficiently samples physically distinct configurations together with their corresponding occurrence counts. Throughout all experiments in this work that employ AGD method, the total number of autoregressive samples is fixed at $1\times10^{12}$. To reduce the computational cost, the local energies are evaluated using the scalable approach proposed in Ref.~\cite{Zhao2023Scalable}.
The complete optimization workflow and pseudocode for the annealed gradient descent procedure are outlined in Algorithm~\ref{alg:agd}. 

To evaluate the performance of neural-network quantum states, we benchmark our approach on a diverse set of molecular Hamiltonians. The Hamiltonians listed in Table~\ref{tab:nqs_crnn_molecules} are taken directly from the supporting code repository associated with Ref.~\cite{Barrett2022}, where they were generated using the following workflow. Molecular geometries were obtained from the PubChem database~\cite{Kim2016pubchem}, and the corresponding electronic-structure calculations were performed in the STO-3G minimal basis using \textsc{Psi4}~\cite{Smith2020psi4}. The resulting molecular integrals were then processed with \textsc{OpenFermion}~\cite{McClean2020openfermion} to construct the second-quantized electronic Hamiltonians, which were subsequently mapped onto qubit Hamiltonians via the Jordan-Wigner transformation. The initial temperature parameters and the number of neurons per layer used to calculate these Hamiltonians using AGD are listed in Table ~\ref{tab:nqs_crnn_molecules}. Except for molecule \ce{LiCl}, which has an annealing step count of 7k steps, the other molecules have an annealing step count of 4k steps. Besides, the results reported in Table~\ref{tab:nqs_crnn_molecules} are obtained with $10^4$ optimization steps for all molecules, except for \ce{LiCl} ($1.5\times10^4$ steps) and \ce{C2} ($2\times10^4$ steps).
To further assess our method, we apply it to the linear hydrogen chain \ce{H10} with an interatomic separation of $2.0$~Bohr in the STO-3G basis, using canonical (restricted Hartree-Fock) molecular orbitals; the hydrogen chain is a prototypical benchmark system for correlated electronic-structure methods, whose stretched geometries probe the strongly correlated regime. The results shown in Fig.~\ref{fig:mol_results}(b) are obtained with a network of 230 hidden units per layer (481{,}164 parameters in total) at a initial temperature of $T = 10.0$ with decay steps 15k and total steps $1.5\times 10^5$.
In all molecular experiments, all network parameters are optimized using the Adam optimizer with decay rates $(\beta_1, \beta_2) = (0.9, 0.99)$ and $\epsilon = 10^{-8}$. We employ parameter-group-dependent learning rates: the parameters of the phase subnetwork are updated with a learning rate of $10\eta$, while the amplitude subnetwork and all remaining parameters use the base learning rate $\eta = 1\times 10^{-3}$. All runs listed in Table ~\ref{tab:nqs_crnn_molecules} are performed with a fixed random seed ($\mathrm{seed} = 111$), and only the particle-number (electron-number) conservation constraint is imposed on the sampled configurations. 

Having established the performance of our approach on molecular Hamiltonians, we next turn to frustrated quantum spin systems and benchmark it on the $J_1$-$J_2$ Heisenberg model in both one and two dimensions, scanning the frustration ratio $J_2/J_1$ across the weakly and strongly frustrated regimes. The hyperparameters used for each value of $J_2/J_1$ are summarized in Table~\ref{tab:nqs_crnn_j1j2}, including the number of hidden units per layer, the base learning rate $\eta$, the initial temperature $T$, the temperature-decay step count, and the total number of optimization steps. For the one-dimensional chain, the weakly frustrated points $J_2/J_1 \in \{0.0, 0.1, 0.2, 0.3, 0.4, 0.5\}$ are obtained with $10^5$ optimization steps using 100 hidden units per layer, a base learning rate of $\eta = 1\times10^{-3}$, an initial temperature of $T = 4.0$, and a decay step count of 8k. In the more strongly frustrated regime, the point $J_2/J_1 = 0.6$ uses the same network and annealing schedule but requires $3\times10^5$ optimization steps, while $J_2/J_1 = 0.8$ additionally employs a larger network of 150 hidden units per layer, a reduced base learning rate of $\eta = 2.5\times10^{-4}$, an initial temperature of $T = 5.0$, and a decay step count of 10k. 

For our AGD-enhanced NQS experiment of two-dimensional square lattice, all points are computed with $3\times10^5$ optimization steps, a base learning rate of $\eta = 2.5\times10^{-4}$, and a decay step count of 30k. The points $J_2/J_1 \in \{0.0, 0.2\}$ use 200 hidden units per layer with an initial temperature of $T = 6.0$; the point $J_2/J_1 = 0.5$, lying in the maximally frustrated region, uses 200 hidden units with $T = 4.0$; and the point $J_2/J_1 = 0.8$ uses an enlarged network of 250 hidden units per layer with $T = 6.0$. In all spin-model experiments, we adopt the same optimization protocol as in the molecular calculations described above. All $J_1$-$J_2$ Heisenberg model runs are performed with a fixed random seed ($\mathrm{seed} = 111$). 
The GRU-NQS data points shown in Fig.~\ref{fig:J1J2_combined} are obtained directly from the $J_1$--$J_2$ implementation provided in the public GitHub repository accompanying the work of Hibat-Allah \emph{et al.}~\cite{HibatAllah2020}. Each data point represents the average over three independent runs with different random seeds, where each run consists of $3\times10^{5}$ optimization steps. We adopt the default initial learning rate of $2.5\times10^{-4}$ and a sample size of $500$, together with the learning-rate decay schedule implemented in the original code, $\eta_{\mathrm{decayed}} = 1/\left(1/\eta + t/10\right)$, where $\eta$ denotes the initial learning rate and $t$ the optimization step. For a fair comparison, the model size and the number of hidden layers of the GRU-NQS are set to be identical to those listed in Table~\ref{tab:nqs_crnn_j1j2}.

Finally, we describe the learning-rate schedules employed during the variational optimization. For the molecular systems, as well as for the one-dimensional $J_1$-$J_2$ model with $J_2/J_1 \le 0.6$, we adopt a two-stage schedule: starting from an initial learning rate $\eta_0 = 1\times10^{-3}$, the learning rate is first annealed following a cosine schedule over the initial $T_{\mathrm{decay}} = 3000$ optimization steps,
\begin{equation}
  \eta(t) = \eta_{\min}
  + \tfrac{1}{2}\left(\eta_0 - \eta_{\min}\right)
  \left[1 + \cos\!\left(\pi t / T_{\mathrm{decay}}\right)\right],
  \qquad t \le T_{\mathrm{decay}},
\end{equation}
with $\eta_{\min} = 5\times10^{-4}$, after which the learning rate is held constant at $\eta_{\min}$ for the remainder of the training. For the more strongly frustrated regime of the one-dimensional model ($J_2/J_1 = 0.8$) and for all two-dimensional $J_1$-$J_2$ calculations, where the optimization landscape is considerably more challenging, we instead employ a smooth inverse-time decay,
\begin{equation}
  \eta(t) = \frac{\eta_0}{1 + \alpha\,\eta_0\, t},
\end{equation}
with $\eta_0 = 2.5\times10^{-4}$ and $\alpha = 0.1$, which provides a slowly and monotonically decreasing learning rate over the full course of the optimization and thereby suppresses stochastic fluctuations in the late stage of the training.

\begin{table}[htbp]
  \centering
  \caption{Hyperparameters used for the one- and two-dimensional
  $J_1$-$J_2$ Heisenberg model calculations: the frustration ratio
  $J_2/J_1$, the number of hidden units per layer, the base learning rate
  $\eta$, the initial temperature $T$, the temperature-decay step count, and
  the total number of optimization steps.}
  \label{tab:nqs_crnn_j1j2}
  \begin{tabular}{ccccccc}
    \hline\hline
    System & $J_2/J_1$ & Hidden units & $\eta$ & $T$ & Decay steps & Total steps \\
    \hline
    1D & $0.0$-$0.5$ & 100 & $1\times10^{-3}$   & $4.0$ & $8\mathrm{k}$  & $1\times10^{5}$ \\
    1D & $0.6$        & 100 & $1\times10^{-3}$   & $4.0$ & $8\mathrm{k}$  & $3\times10^{5}$ \\
    1D & $0.8$        & 150 & $2.5\times10^{-4}$ & $5.0$ & $10\mathrm{k}$ & $3\times10^{5}$ \\
    \hline
    2D & $0.0$        & 200 & $2.5\times10^{-4}$ & $6.0$ & $30\mathrm{k}$ & $3\times10^{5}$ \\
    2D & $0.2$        & 200 & $2.5\times10^{-4}$ & $6.0$ & $30\mathrm{k}$ & $3\times10^{5}$ \\
    2D & $0.5$        & 200 & $2.5\times10^{-4}$ & $4.0$ & $30\mathrm{k}$ & $3\times10^{5}$ \\
    2D & $0.8$        & 250 & $2.5\times10^{-4}$ & $6.0$ & $30\mathrm{k}$ & $3\times10^{5}$ \\
    \hline\hline
  \end{tabular}
\end{table}

\begin{algorithm}[H]
\nolinenumbers
\caption{Annealed Gradient Descent for NQS}
\label{alg:agd}
\SetKwData{Left}{left}\SetKwData{This}{this}\SetKwData{Up}{up}
\SetKwFunction{Union}{Union}\SetKwFunction{FindCompress}{FindCompress}
\SetKwInOut{Input}{input}\SetKwInOut{Output}{output}
\KwIn{Hamiltonian $\hat{H}$, $T_{\mathrm{init}}$, $T_{\mathrm{decay}}$, $N_{\mathrm{steps}}$, $N_S$}
\KwOut{Optimized $\bm{\theta}$}
Initialize $\bm{\theta}$\;
\For{$t = 0$ to $N_{\mathrm{steps}}-1$}
{
    Sample $\{\sigma_i\}_{i=1}^{N_S} \sim p_{\bm{\theta}}(\sigma)$\;
    Compute $E_{\mathrm{loc}}(\sigma_i) = \langle\sigma_i|\hat{H}|\psi_{\bm{\theta}}\rangle/\langle\sigma_i|\psi_{\bm{\theta}}\rangle$\;
    $T(t) \leftarrow T_{\mathrm{init}} - (T_{\mathrm{init}}-1)\cdot\min(1, t/T_{\mathrm{decay}})$\;
    $w_i \propto \exp\left[\frac{2}{T(t)}\log|\psi_{\bm{\theta}}(\sigma_i)|\right]$\;
    $E_{\mathrm{baseline}} \leftarrow \sum_i w_i E_{\mathrm{loc}}(\sigma_i)$\;
    $\mathcal{L} \leftarrow 2\,\mathrm{Re}\left[ \sum_i w_i \left( E_{\mathrm{loc}}(\sigma_i) - E_{\mathrm{baseline}} \right) \log\psi_{\bm{\theta}}^*(\sigma_i) \right]$;
    $\bm{\theta} \leftarrow \mathrm{Adam}(\bm{\theta}, \nabla_{\bm{\theta}}\mathcal{L})$\;
}
\Return{$\bm{\theta}$}\;
\end{algorithm}

\section{Gradient Expression}
\label{sec:gradient}
To understand the specific meaning of the gradient formula, we can use the expansion method to see it clearly. This expansion approximation reveals how the probability current generated by the neural network flows. Let $\mathcal{L}=\langle H\rangle_{\psi_{\bm{\theta}}}$ be variational energy, $\Delta E(\sigma)= E_{loc}(\sigma) - \langle E_{loc}\rangle $ be the deviation of the local energy and $\hat{\grads}$ be the annealed gradient. Instead of estimating the gradients via sampling configurations according to the $p_{\bm{\theta}}=|\psi_{\bm{\theta}}|^2$, we replace $p_{\bm{\theta}}$ with $q_{\bm{\theta}}(T) = p_{\bm{\theta}}^{\frac{1}{T}}$ and consider $T=1+\epsilon$ where $\epsilon \ll 1$. Then, we have

\begin{align}
    \partial_{{\bm{\theta}}}\mathcal{L}_{T}(\bm{\theta})
    &\propto 2\operatorname{Re} \left\{ \sum_\sigma \left[ p_{\bm{\theta}}(\sigma) \right]^{\frac{1}{1+\epsilon}} \Delta E(\sigma)\partial_{\bm{\theta}}\log \psi^*_{\bm{\theta}}(\sigma) \right\}  
    \label{eq:modified gradient Line1}
    \\
    &\simeq 2\operatorname{Re}\left\{\sum_\sigma p_{\bm{\theta}}(\sigma) \bigl( 1 - \epsilon \log p_{\bm{\theta}}(\sigma) \bigr) \Delta E(\sigma) \partial_{\bm{\theta}}\log \psi^*_{\bm{\theta}}(\sigma) \right\} 
    \label{eq:modified gradient Line2}
    \\
    &= \underbrace{2\operatorname{Re}\left\{ \sum_\sigma \left[ p_{\bm{\theta}}(\sigma) \right]  \Delta E(\sigma)\partial_{\bm{\theta}}\log \psi^*_{\bm{\theta}}(\sigma)\right\}}_{\partial_{\bm{\theta}}\mathcal{L}(\bm{\theta})} + 
    2\operatorname{Re} \left\{\epsilon\sum_\sigma \left[ -p_{\bm{\theta}}(\sigma) \log p_{\bm{\theta}}(\sigma) \right]  \Delta E(\sigma)\partial_{\bm{\theta}}\log \psi^*_{\bm{\theta}}(\sigma)\right\} \label{eq:modified gradient Line3},  
\end{align}
where Eq.~\ref{eq:modified gradient Line2} follows from a first-order Taylor expansion of $q_{\bm{\theta}}$ in $\epsilon$,
\begin{equation}
    p^{\frac{1}{1+\epsilon}}
    = p\left[1-\epsilon\log p+\epsilon^2\left(\log p+\frac{(\log p)^2}{2}\right)+\mathcal{O}(\epsilon^3)\right],
    \label{equ:Taylor series}
\end{equation}
Thus, we define $\partial_{\bm{\theta}} \mathcal{K}({\bm{\theta}})$ as a general form of the each term in Eq.\ref{eq:modified gradient Line3},
\begin{equation}
    \partial_{\bm{\theta}} \mathcal{K}({\bm{\theta}})= 2\operatorname{Re}\left\{ \sum_\sigma w(p_{\bm{\theta}}(\sigma))\Delta E(\sigma)\partial_{\bm{\theta}}\log \psi^*_{\bm{\theta}}(\sigma)\right\}.
    \label{equ:update_para}
\end{equation}
Thus, for instance, the vanilla gradient $\partial_{\bm{\theta}}\mathcal{L}(\bm{\theta})$ and the second term in Eq.~\ref{eq:modified gradient Line3} can be naturally written as $\partial_{\bm{\theta}} \mathcal{K}({\bm{\theta}})$ with $w(p)=p$ and $w(p)\equiv p\log{\frac{1}{p}}=-p\log p$, respectively.

We then consider the contribution with respect to each terms in Eq.~\ref{eq:modified gradient Line3} for the change of $\bm{\theta}$,
\begin{align}
\Delta \bm{\theta} &= -\eta\partial_{\bm{\theta}} \mathcal{K}(\bm{\theta}),\\
&= - 2\eta \operatorname{Re}\left\{ \sum_\sigma w(p_{\bm{\theta}}(\sigma))\Delta E(\sigma)\partial_{\bm{\theta}}\log \psi^*_{\bm{\theta}}(\sigma)\right\},\label{eq:delta_theta_2}\\
&= - 2\eta \sum_\sigma w(p_{\bm{\theta}}(\sigma))\Delta E(\sigma)\partial_{\bm{\theta}}\log p_{\bm{\theta}}(\sigma).\label{eq:delta_theta_3}
\end{align}
From Eq.\ref{eq:delta_theta_2} to Eq.\ref{eq:delta_theta_3}, it is based on the fact that $\log\psi^*_{\bm{\theta}}(\sigma)=\log p_{\bm{\theta}} - i\phi_{\bm{\theta}}$ where $\phi_{\bm{\theta}}$ is the phase of NQS.
For an arbitrary single configuration $\sigma_0$, the change of $\log p_{\bm{\theta}}$ can be computed as 
\begin{equation}
    \Delta \log p_{\bm{\theta}} (\sigma_0)\approx \big(\partial_{\bm{\theta}} \log p_{\bm{\theta}}(\sigma_0)\big)^{\top}\Delta{\bm{\theta}}.\label{eq:delta_logp}
\end{equation}
then, substituting Eq.\ref{eq:delta_theta_3} into Eq.\ref{eq:delta_logp}, we have
\begin{equation}
    \Delta \log p_{\bm{\theta}}(\sigma_0) \approx - 2\eta \sum_\sigma w(p_{\bm{\theta}}(\sigma))\Delta E(\sigma)\big(\partial_{\bm{\theta}} \log p_{\bm{\theta}}(\sigma_0)\big)^{\top}\partial_{\bm{\theta}}\log p_{\bm{\theta}}(\sigma).
    \label{eq:appendix_delta_logp}
\end{equation}
When we only consider the contribution of a single configuration, i.e. $\{\sigma\}=\{\sigma_0\}$, then Eq.\ref{eq:appendix_delta_logp} becomes,
\begin{equation}
    -2\eta \cdot w(p_{\bm{\theta}}(\sigma_0))\Delta E(\sigma_0) \big| \partial_{\bm{\theta}} \log p_{\bm{\theta}}(\sigma_0) \big|^2.
\end{equation}

We can observe that when $\Delta E(\sigma_0)<0$ which makes $\Delta \log p_{\bm{\theta}}(\sigma_0)>0$ and leads increase of $\log p_\theta(\sigma)$, thereby increasing the probability to sample $\sigma_0$.
The effect of the annealing factor on the sampling is illustrated in Fig.~\ref{fig:wp}. The orange dashed line corresponds to the vanilla gradient weight, $w(p)=p$, whereas the blue curve represents the second correction term in Eq.~\ref{eq:modified gradient Line3}, with the entropy-like weighting function $w(p)=-p\log p$. In the low-probability region, particularly for $p\in[0,0.1]$, this entropy-like weight amplifies the contribution of configurations $\sigma$ with small model probability $p_{\bm{\theta}}(\sigma)$. As a result, physically important but underrepresented configurations are less likely to be missed during gradient estimation. At the same time, $w(p)=-p\log p$ vanishes as $p\to 1$, so it does not further enhance configurations that already have high probability. This prevents dominant configurations from overwhelming the update and helps the NQS maintain access to low-probability but relevant regions of the Hilbert space.

\begin{figure}[H]
  \centering
  \includegraphics[width=.45\linewidth]{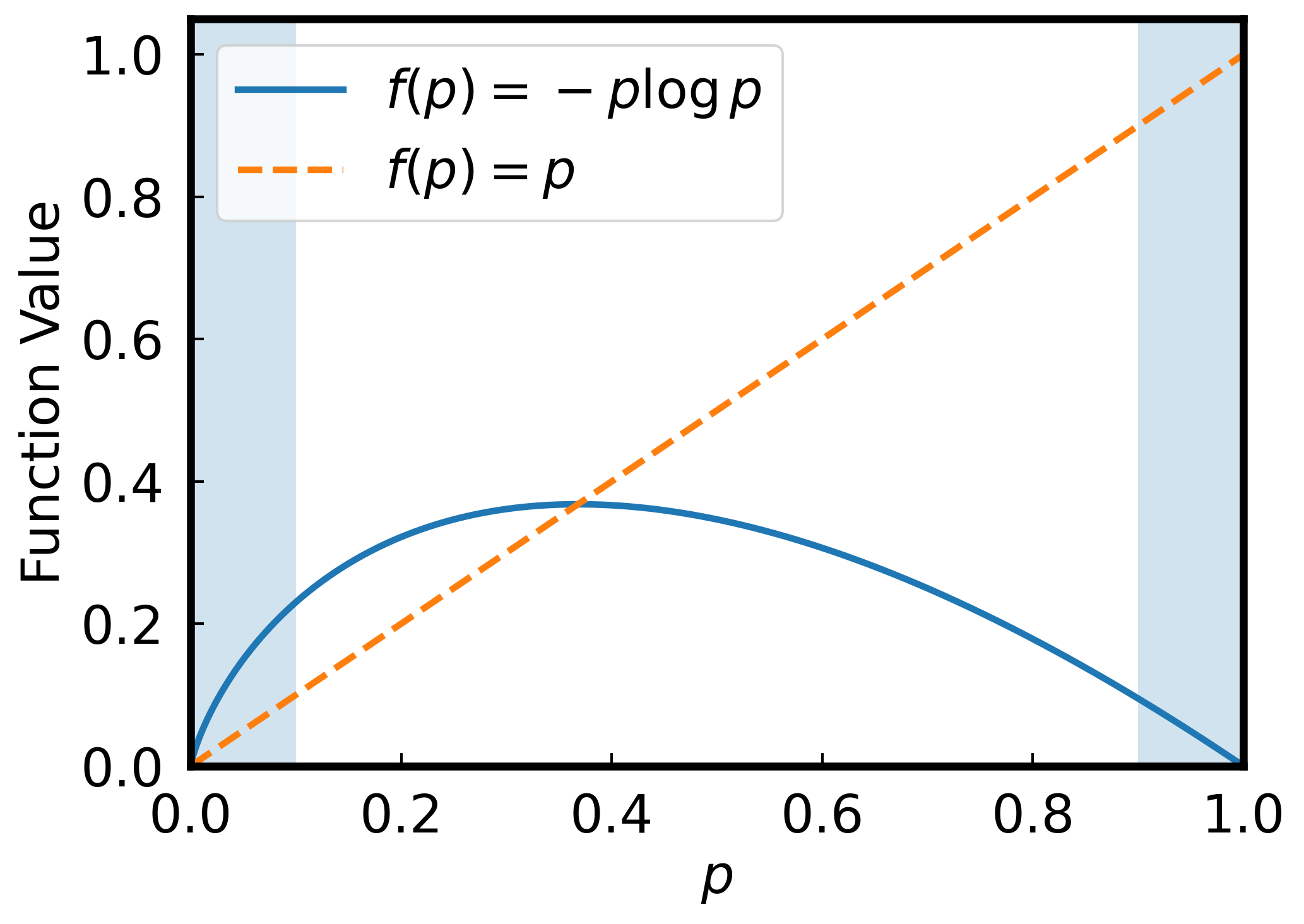}
    \caption{Comparison between the entropy-like weighting function $f(p)=-p\log p$ and the linear weighting function $f(p)=p$ over $p\in[0,1]$.}
    \label{fig:wp}
\end{figure}

Compared to the zeroth-order weight $w(p)=p$, which corresponds to standard gradient descent, the higher-order weight exhibit a logarithmic enhancement in the regime of small $p$. Specifically, as indicated by Eq.~\ref{equ:Taylor series}, the $n$-th order contribution provides an enhancement scaling as $(\log p)^n$ as shown in Fig.\ref{fig:series}.

\begin{figure}[H]
  \centering
    \includegraphics[width=0.6\linewidth]{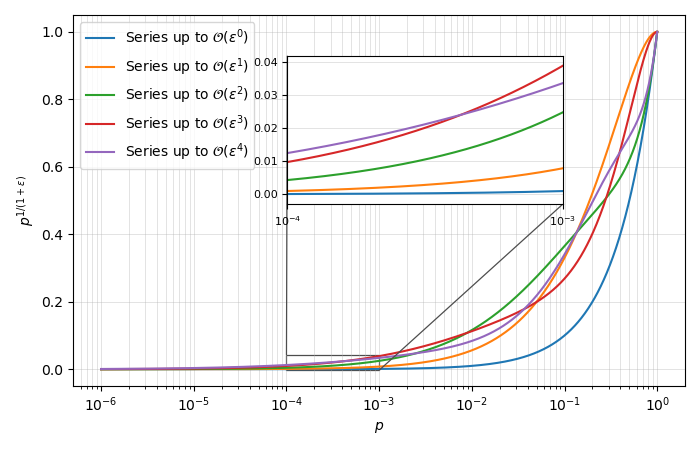}
  \caption{Comparison of \(p^{1/(1+\epsilon)}\) and its Taylor expansions in \(\epsilon\) evaluated at \(\epsilon=1\). 
  Relative to the zeroth-order approximation \(\mathcal{O}(\epsilon^{0})\), the higher-order truncations that incorporate \(\epsilon\)-dependent corrections exhibit a logarithmically enhanced weighting in the small-\(p\) regime.
  }
  \label{fig:series}
\end{figure}

\section{Proof of Theorem~\ref{thm:Hsubspace}}
\label{Appendix:proof of thm:Hsubspace}
When the probability of a configuration $\sigma_0$ is misestimated by the neural network and becomes so strongly suppressed that the optimization fails to sample $\sigma_0$ over many steps, the sampling process can be regarded as effectively unable to reach this configuration. The only remaining channel through which $\sigma_0$ can re-enter the local-energy estimate is then through a sampleable configuration $\sigma$ that is coupled to $\sigma_0$ via the matrix element $H_{\sigma\sigma_0}$. Unfortunately, this contribution is suppressed as well. If this suppression is likewise assumed to be negligibly small, the local-energy estimate becomes entirely independent of $\sigma_0$. We now
make this statement precise and present the proof of Theorem~\ref{thm:Hsubspace}.
\begin{proof}
Let $\mathcal{H}$ denote the full Hilbert space and $\mathcal{S}$ be the truncated subspace excluding the unsampled state $|\sigma_0\rangle$. 
Note that this can occur if a component is suppressed due to the sign problem, rendering it inaccessible to sampling. We define the projection operator into this subspace as $\hat{P} = \hat{I} - |\sigma_0\rangle\langle\sigma_0|$. Under the assumption that $|\sigma_0\rangle$ is never sampled , the variational wave function is strictly confined to $\mathcal{S}$, satisfying $|\psi_\theta\rangle = \hat{P}|\psi_\theta\rangle$.

In the framework of VMC, the empirical local energy for any sampled state $|\sigma\rangle \in \mathcal{S}$ is given by
\begin{equation}
    E_{\mathrm{loc}}(\sigma) = \sum_{\sigma^\prime \neq \sigma_0} H_{\sigma \sigma^\prime} \frac{\psi_\theta(\sigma^\prime)}{\psi_\theta(\sigma)} + H_{\sigma \sigma_0}\frac{\psi_\theta(\sigma_0)}{\psi_\theta(\sigma)}.
    \label{equ:S_local_energy}
\end{equation}
where, since $\sigma_0$ is never sampled, we have $\psi_\theta (\sigma_0) \propto p_\theta(\sigma_0) \approx 0$, allowing us to neglect the second term. Strictly speaking, a near-vanishing function value does not necessarily imply a vanishing derivative. However, as explicitly shown in Eq.~\ref{eq:vmc_gradient}, the gradient is fundamentally governed by the parameter derivatives of $\psi_{\bm{\theta}}$, with $E_{\mathrm{loc}}$ acting merely as a multiplicative weight that does not enter the differentiation process. Furthermore, the gradient estimator is evaluated as a sum over the sampled configurations.
Thus this restricted summation is mathematically equivalent to defining an effective Hamiltonian $\hat{H}_{\mathrm{subspace}} = \hat{P}\hat{H}\hat{P}$. Consequently, for any $\sigma \neq \sigma_0$, the local energy can be rigorously expressed as
\begin{equation}
    E_{\mathrm{loc}}(\sigma) = \frac{\langle\sigma| \hat{H}_{\mathrm{subspace}} |\psi_\theta\rangle}{\langle\sigma|\psi_\theta\rangle}.
    \label{equ:H_eff}
\end{equation}

The objective of gradient-based optimization is to minimize the energy functional $E(\theta) = \langle\psi_\theta| \hat{H}_{\mathrm{subspace}} |\psi_\theta\rangle / \langle\psi_\theta|\psi_\theta\rangle$ within the restricted subspace $\mathcal{S}$. Assuming the NQS possesses sufficient expressivity, convergence to a stationary point dictates that the wave function must satisfy the variational extremum condition:
\begin{equation}
    \hat{H}_{\mathrm{subspace}} |\psi_\theta\rangle = E_{0} |\psi_\theta\rangle.
\end{equation}
This implies that the optimization algorithm invariably drives $|\psi_\theta\rangle$ to become an exact eigenstate of the effective Hamiltonian  $\hat{H}_{\mathrm{subspace}}$.

Substituting this eigenvalue equation back into the definition of the local energy Eq.~\ref{equ:H_eff}  yields
\begin{equation}
    E_{\mathrm{loc}}(\sigma) = \frac{\langle\sigma| (E_0 |\psi_\theta\rangle)}{\langle\sigma|\psi_\theta\rangle} = E_{0} \frac{\langle\sigma|\psi_\theta\rangle}{\langle\sigma|\psi_\theta\rangle} = E_{0}.
\end{equation}
Therefore, for all sampled states $\sigma \in \mathcal{S}$, the residual energy identically vanishes, yielding $\Delta E(\sigma) \equiv E_{\mathrm{loc}}(\sigma) - E \to 0$. Due to Eq.~\ref{eq:appendix_delta_logp}, we can concluded $\Delta \log p_\theta(\sigma_0) \to 0$.
\end{proof}

Following Theorem~\ref{thm:Hsubspace}, $\Delta \log p_\theta(\sigma_0)$ vanishes. This demonstrates that if $p(\sigma_0) \to 0$ due to any mechanism--such as suppression induced by the sign problem--the subsequent updates to $p(\sigma_0)$ will also tend to zero. Consequently, the neural network becomes trapped in a local minimum. However, the aforementioned assumption is highly idealized. In practice, given a sufficiently large sample size, a specific component may still be sampled a handful of times, even under severe suppression induced by the sign problem. If the network manages to dynamically correct the sign structure during training, the local energy deviation $\Delta E(\sigma)$ will not vanish. Instead, this recovery provides crucial gradient information, driving the energy of the NQS significantly lower.

\end{document}